\pdfoutput=1
\documentclass[a4paper,UKenglish,11pt]{amsart}

\usepackage[utf8]{inputenc}
\usepackage{lmodern}
\usepackage[T1]{fontenc}
\usepackage[margin=1in]{geometry}

\usepackage{mathtools}

\usepackage{thm-restate}
\usepackage{amsthm,amsfonts,amssymb,mathrsfs}
\usepackage{textcomp}

\usepackage{pgf,interval}
\intervalconfig{
  soft open fences,
  oo/.style={open},
  oc/.style={open left},
  co/.style={open right},
  cc/.style={},
}

\usepackage[noend]{algpseudocode}
\usepackage{algorithm}
\newcommand*{\Let}{\textbf{let}}
\newcommand*{\Increment}{\textbf{increment}}
\makeatletter
\newcommand{\algmargin}{\the\ALG@thistlm}
\makeatother
\newlength{\whilewidth}
\algdef{SE}[parWHILE]{parWhile}{EndparWhile}[1]
  {\parbox[t]{\dimexpr\linewidth-\algmargin}{%
     \hangindent\whilewidth\strut\algorithmicwhile\ #1\ \algorithmicdo\strut}}{\algorithmicend\ \algorithmicwhile}%
\algnewcommand{\parState}[1]{\State%
  \parbox[t]{\dimexpr\linewidth-\algmargin}{\strut #1\strut}}

\usepackage[shortlabels]{enumitem}
\newlist{enumerata}{enumerate}{1}
\setlist[enumerate]{label=\upshape{(\roman*)}}
\setlist[enumerata]{label=\upshape{(\alph*)}}
\setlist[description]{labelindent=\parindent}
\makeatletter
\def\namedlabel#1#2{\begingroup
    #2%
    \def\@currentlabel{#2}%
    \phantomsection\label{#1}\endgroup
}
\makeatother

\usepackage{bookmark}
\usepackage[capitalise]{cleveref}
\usepackage{microtype}

\makeatletter
\AtBeginDocument{
  \let\thepdftitle\@title
  \let\thepdfauthors\shortauthors
  \nxandlist{, }{ and }{, and }\thepdfauthors
  \hypersetup{
    pdftitle = {\thepdftitle},
    pdfauthor = {\thepdfauthors}
  }
}
\makeatother

\newtheorem{theorem}{Theorem}
\newtheorem{lemma}[theorem]{Lemma}
\newtheorem{proposition}[theorem]{Proposition}

\theoremstyle{definition}
\newtheorem{definition}[theorem]{Definition}

\newcommand{\EE}{\mathbb{E}}

\newcommand{\cA}{\mathcal{A}}
\newcommand{\cC}{\mathcal{C}}
\newcommand{\bF}{\mathbf{F}}
\newcommand{\cF}{\mathcal{F}}
\newcommand{\sH}{\mathscr{H}}
\newcommand{\bI}{\mathbf{I}}
\newcommand{\cI}{\mathcal{I}}
\newcommand{\bJ}{\mathbf{J}}
\newcommand{\cP}{\mathcal{P}}
\newcommand{\bS}{\mathbf{S}}
\newcommand{\bU}{\mathbf{U}}
\newcommand{\bX}{\mathbf{X}}
\newcommand{\bY}{\mathbf{Y}}

\newcommand{\Lam}{\Lambda}
\newcommand{\lam}{\lambda}
\newcommand{\Gam}{\Gamma}
\newcommand{\gam}{\gamma}
\renewcommand{\epsilon}{\varepsilon}
\newcommand{\eps}{\varepsilon}
\newcommand*{\blank}{\mathfrak{B}}

\DeclareMathOperator{\bla}{bla}
\DeclareMathOperator{\col}{col}
\DeclareMathOperator{\In}{In}
\DeclareMathOperator{\ind}{ind}
\DeclareMathOperator{\mad}{mad}

\DeclareMathOperator{\unc}{unc}

\newcommand*{\fix}{\mathrm{\normalfont\textsc{AddressB}}}
\newcommand*{\remove}{\mathrm{\normalfont\textsc{Remove}}}
\newcommand*{\sample}{\mathrm{\normalfont\textsc{Sample}}}

\title{An algorithmic framework for colouring locally sparse graphs}

\author[E.\ Davies]{Ewan Davies}
\address{Department of Computer Science, University of Colorado Boulder, USA}
\email{maths@ewandavies.org}
\thanks{(E.\ Davies) The research leading to these results has received funding from the European Research Council under the European Union's Seventh Framework Programme (FP7/2007-2013) / ERC grant agreement \textnumero{} 339109. Part of this work was done while the author was visiting the Simons Institute for the Theory of Computing.}

\author[R.\ J.\ Kang]{Ross J. Kang}
\address{Department of Mathematics, Radboud University Nijmegen, Netherlands.}
\email{ross.kang@gmail.com}
\thanks{(R.\ J.\ Kang) Supported by a Vidi grant (639.032.614) of the Netherlands Organisation for Scientific Research (NWO)}

\author[F.\ Pirot]{Fran\c{c}ois Pirot}
\address{G-SCOP, CNRS, Univ. Grenoble Alpes, Grenoble, France}
\email{francois.pirot@grenoble-inp.fr}

\author[J.-S.\ Sereni]{Jean-S\'{e}bastien Sereni}
\address{Service Public Fran\c{c}ais de la Recherche, Centre National de la Recherche Scientifique, CSTB (ICube), Strasbourg, France}
\email{sereni@kam.mff.cuni.cz}

\begin{document}

\begin{abstract}
We develop an algorithmic framework for graph colouring that reduces the problem to verifying a local probabilistic property of the independent sets.

With this we give, for any fixed $k\ge 3$ and $\varepsilon>0$, a randomised polynomial-time algorithm for colouring graphs of maximum degree $\Delta$ in which each vertex is contained in at most $t$ copies of a cycle of length $k$, where $1/2\le t\le \Delta^\frac{2\varepsilon}{1+2\eps}/(\log\Delta)^2$, with $\lfloor(1+\varepsilon)\Delta/\log(\Delta/\sqrt t)\rfloor$ colours.  

This generalises and improves upon several notable results including those of Kim (1995) and Alon, Krivelevich and Sudakov (1999), and more recent ones of Molloy (2019) and Achlioptas, Iliopoulos and Sinclair (2019). This bound on the chromatic number is tight up to an asymptotic factor $2$ and it coincides with a famous algorithmic barrier to colouring random graphs.
\end{abstract}

\maketitle


\section{Introduction}\label{sec:intro}

Let $G=(V,E)$ be a graph. An \emph{independent set} of $G$ is a vertex subset that induces an edgeless subgraph of $G$. The \emph{independence number}~$\alpha(G)$ of $G$ is the cardinality of a largest independent set of $G$. The \emph{chromatic number}~$\chi(G)$ of $G$ is the least number of parts in a partition of $V$ into independent sets of $G$.
Determining or bounding these structural parameters have been of fundamental importance to algorithms, optimisation, and operations research~\cite{Kar72}. Moreover, they have been central in the development of combinatorial mathematics, especially with respect to random graphs and extremal combinatorics~\cite{Ram29,ErSz35,Erd47,Erd59}.

The algorithmic and combinatorial perspectives are inextricably linked.
As an example, with an interpretation of the random graph $G_{n,1/2}$ as a model of average-case behaviour, Karp asked in 1976~\cite{Kar76} if for some positive~$\eps$ there is a polynomial-time algorithm that outputs an independent set in $G_{n,1/2}$ of size $(1+\eps)\log_2 n$ with probability tending to~$1$ as the number of vertices~$n$ tends to infinity, that is, \emph{with high probability (w.h.p.)}. (It is a basic fact that existentially we have $\alpha(G_{n,1/2}) \sim 2\log_2 n$ w.h.p.) Karp's question remains open and has helped to provoke an influential, sustained series of investigations in random graph theory, cf.~e.g.~\cite{KaMc15} for a survey from the perspective of graph colouring.
As another example, there are notorious gaps between the best-known upper and lower estimates on classical Ramsey numbers, but bounds have nevertheless proven useful towards approximation algorithms, cf.~e.g.~\cite{HaRa94,BGG15}.

Our main contribution is a novel framework for the asymptotic global structure---in terms of independent sets or colourings---of graphs that satisfy some local sparsity condition, having e.g.\ few edges in any induced neighbourhood subgraph. This framework is built around the establishment of elementary local properties of the so-called \emph{hard-core model} on a graph, a probabilistic approach having its roots in statistical physics. Our work lies near the interface between the above-mentioned parallel perspectives, and in fact is closely related to the two examples above. In this extended abstract we focus on algorithmic aspects of our framework, through one specific (and important) application, and show a comfortable incorporation of modern stochastic local search machinery to improve on the state of the art. In a companion paper~\cite{DKPS20main} we explore a broader but also more combinatorial array of applications, prioritising not-necessarily-algorithmic existential results.

One old and basic starting point for this research is the pursuit of global asymptotic structure in triangle-free graphs, that is, in graphs having no edges whatsoever in any induced neighbourhood subgraph. The search for large independent sets in this context corresponds to the classic off-diagonal case of Ramsey numbers~\cite{Ram29,ErSz35,AKS80,AKS81,She83,Kim95a,Boh09,BoKe13+,FGM20,DJPR18}, a foundational and profoundly difficult problem in combinatorics. The search for good colourings in this context, a related but more delicate task, is also an important challenge of classic origins, cf.~\cite{Zyk49,UD54}.

There is particular interest in graphs of bounded maximum degree, with natural links to approximation algorithms, cf.~e.g.~\cite{HaRa94,BGG15}. For colouring this interest originated in a question of Vizing from 1968~\cite{Viz68}: what is the largest chromatic number taken over all triangle-free graphs of maximum degree $\Delta$?
(Even without the triangle-free condition a trivial greedy argument yields an upper bound of $\Delta+1$, which is sharp for odd cycles and cliques.)
Simultaneously strengthening a seminal result of Ajtai, Koml\'{o}s and Szemer\'{e}di~\cite{AKS81} for the independence number and answering Vizing's question up to the choice of leading asymptotic constant, Johansson~\cite{Joh96} devised a sophisticated semirandom colouring procedure to establish an upper bound of $O(\Delta/\log \Delta)$ as $\Delta\to\infty$.  
Recently, in a dramatic advance, Molloy~\cite{Mol19} employed \emph{entropy compression} for a simplified proof and an intriguing improvement over Johansson's result, quantitatively matching an analogous independence number bound of Shearer~\cite{She83}.

\begin{theorem}[Molloy~\cite{Mol19}]\label{thm:molloy}
  For all $\eps > 0$, there exists $\Delta_0$ such that if $\Delta \ge \Delta_0$, then $\chi(G)\le (1+\eps)\Delta/\log\Delta$ for any given triangle-free graph $G$ of maximum degree $\Delta$.
  There is a randomised algorithm that in polynomial time w.h.p.~constructs a certificate colouring of $G$.
\end{theorem}

\noindent
For a hint of how difficult it might be to improve on this result, particularly with respect to the asymptotic leading constant of $1$, one can take two issues into consideration.
First, lowering the constant appreciably would by the same token improve upon the best to date lower bounds on the classical off-diagonal Ramsey numbers (which are due to Shearer~\cite{She83} as alluded to above), and it would constitute a breakthrough in quantitative Ramsey theory.
Second, a lowering of the constant \emph{and} with a polynomial-time algorithm would essentially imply a positive answer to the direct analogue of Karp's question above, for the random $\Delta$-regular rather than binomial random graph. Indeed, the following result is well known in random graph theory.

\begin{proposition}
\label{prop:regular}
For all $\eps>0$, there exists some $\Delta_0$ such that for all fixed $\Delta\ge\Delta_0$, we have the following for all $n$ sufficiently large. With probability at least $1-\eps$, the random $\Delta$-regular graph~$G_{n,\Delta}$ on $n$ vertices is triangle-free and satisfies $\alpha(G_{n,\Delta}) \in (2\pm\eps)(n\log \Delta)/\Delta$.
\end{proposition}

\noindent
Since $\alpha(G) \ge |V|/\chi(G)$ for all $G=(V,E)$,
this shows the asymptotic leading term in Molloy's result (and the corresponding result of Shearer) to be correct up to a factor $2$.

We offer a more general principle behind \cref{thm:molloy}, through locally-defined probabilistic properties of the independent sets. 
Through this, one may witness that certain methods behind \cref{thm:molloy} are sharp and cannot be improved asymptotically; we discuss this in Subsection~\ref{sub:optimality}.
Important too is that the principle is flexible enough for a host of applications, which we partially present through this extended abstract (with more treated in the companion paper~\cite{DKPS20main}).
To give a first flavour of the extra breadth in our approach, here is a prototypical version of our main result in this extended abstract. 
For~$k\ge3$, let us define the \emph{fan} $F_k$ of order~$k$ as the graph formed from a path on $k-1$ vertices by adding a vertex joined to all vertices of the path.
We call a graph \emph{$F_k$-free} if it does not contain the fan~$F_k$ as a subgraph.

\begin{theorem}\label{thm:Ckfree}
Fix an integer $k\ge 3$. For all $\eps>0$, there exists some $\Delta_0$ such that if $\Delta\ge\Delta_0$, then $\chi(G)\le (1+\eps)\Delta/\log\Delta$ for any given $F_k$-free graph $G$ of maximum degree $\Delta$.
  There is a randomised algorithm that in polynomial time w.h.p.~constructs a certificate colouring of $G$.
\end{theorem}

\noindent
Note $F_k$ contains a cycle of each length between $3$ and $k$, and so this strengthens \cref{thm:molloy} in a natural way.
For $k>3$, earlier work in this direction~\cite{AKS99,Vu02,AIS19} was not enough to obtain a leading asymptotic constant of $1$ (even without demanding a polynomial-time algorithm).
Keeping in mind \cref{prop:regular}, this constant is at most twice the optimal value, just as for \cref{thm:molloy}.

Another basic but more modern starting point for this research is the investigation of stochastic local search algorithms. In broad terms, given a state space equipped with a probability measure that has designated \emph{flawed subsets} (or \emph{flaws}), under what circumstances is there an efficient randomised algorithm, performing \emph{local} moves, to arrive at a flawless state? (One can think of a flawless state as, say, a satisfying assignment or a colouring.)
In a remarkable breakthrough, Moser~\cite{Mos09} (cf.~\cite{MT10}), showed that the Lov\'{a}sz local lemma~\cite{EL75}---a fundamental result for proving the existence of combinatorial structures with the probabilistic method---follows from an elementary stochastic search algorithm based on resampling parts of the current state. In his analysis, Moser devised the entropy compression method mentioned earlier, and this has since found wide applicability to various search algorithms that backtrack to avoid problematic regions of the state space, cf.~e.g.~\cite{EsPa13}.
Achlioptas, Iliopoulos, and Sinclair~\cite{AIS19} recently gave a powerful algorithmic form of the local lemma that permits the analysis of \emph{hybrid} algorithms, that can both resample and backtrack.
As their main application, they gave the following generalisation of \cref{thm:molloy}, under a smooth relaxation of the triangle-free condition.

\begin{theorem}[Achlioptas, Iliopoulos, and Sinclair~\cite{AIS19}]\label{thm:AIS}
    For all $\eps > 0$, there exists $\Delta_0$ such that if $\Delta \ge \Delta_0$ and $1/2\le t\le \Delta^{\frac{2\eps}{1+2\eps}}/{(\log\Delta)}^2$, then $\chi(G)\le (1+\eps)\Delta/\log(\Delta/\sqrt t)$ for any given graph~$G$ of maximum degree $\Delta$ where each vertex of $G$ is contained in at most $t$
  triangles.
  There is a randomised algorithm that in polynomial time w.h.p.~constructs a certificate colouring of $G$.
\end{theorem}

\subsection{Our contributions}

As mentioned above, our main achievement is the development of a general framework for global graph structure that significantly strengthens the results stated above (Theorems~\ref{thm:molloy},~\ref{thm:Ckfree}, and~\ref{thm:AIS}). It encompasses or improves upon a long line of earlier work in this area~\cite{AIS19,AKS99,Ber19,DJKP18,DJKP18a,Joh96,Kim95,Mol19,Vu02}.
The framework in general reduces the main task to the verification of a probabilistic property of the independent sets that we call \emph{local occupancy}. In several applications this verification is straightforward, resulting in simplified proofs for existing results with matching or improved bounds, cf.~\cite{DKPS20main}. 
Moreover, subject to mild extra conditions, we can give polynomial-time constructions, which is our focus here.
Our main application is a common generalisation of Theorems~\ref{thm:Ckfree} and~\ref{thm:AIS} (and it implies the announced result since $C_k\subset F_k$).

\begin{theorem}\label{thm:Ckapplication}
    Fix an integer~$k\ge3$. For all $\eps > 0$, there exists $\Delta_0$ such that if $\Delta \ge \Delta_0$ and $1/2\le t\le \Delta^{\frac{2\eps}{1+2\eps}}/{(\log\Delta)}^2$, then $\chi(G)\le (1+\eps)\Delta/\log(\Delta/\sqrt t)$ for any given graph $G$ of maximum degree $\Delta$ where each vertex of $G$ is contained in at most $t$
 copies of the fan $F_k$.
  There is a randomised algorithm that in polynomial time w.h.p.~constructs a certificate colouring of $G$.
\end{theorem}

Our method builds upon Molloy's proof of \cref{thm:molloy}, starting with a `blank' partial colouring and resampling the colours of neighbourhoods until a flawless partial colouring is found.
We have distilled the graph structure necessary for this resampling to eventually succeed, namely local occupancy, and this strategy alone suffices for the \emph{existence} of colourings as guaranteed by \cref{thm:Ckapplication}.
Resampling is performed according to the hard-core model and in general this is not known to be possible in polynomial time.
A crucial innovation we develop to raise~$k$ from~$3$ (\cref{thm:AIS}) to an arbitrary integer is an efficient resampling if the neighbourhood contains no long path (\cref{thm:sampling}).
These ideas already suffice for \cref{thm:Ckfree} (which corresponds to the case $t<1$),
but for \cref{thm:Ckapplication} we incorporate an adaptation of the backtracking steps used in~\cite{AIS19} for \cref{thm:AIS} with the stated upper bound on~$t$.
Thus we handle few copies of~$F_k$ by `removing' an edge in each copy
and show that the removed edges are unlikely to stall the algorithm.
That is, we can successfully backtrack away from colourings that make the removed edges monochromatic and still show the algorithm will terminate.
We remark that the upper bound condition on $t$ in Theorems~\ref{thm:AIS} and~\ref{thm:Ckapplication} stems from the demand for a polynomial-time construction.
For existence alone, one can not only substantially relax the condition on $t$ but also allow $k$ to increase as a modest function of $\Delta$, as we show in the companion paper~\cite{DKPS20main}.

\subsection{Optimality}\label{sub:optimality}

As intimated earlier, two important and related facts support the idea that our method is optimal.
The first concerns Karp's longstanding question mentioned earlier. The difficulty in this is now recognised as deriving from \emph{shattering} (or \emph{dynamic replica
symmetry breaking}, as it is referred to in statistical physics) in the collection of independent sets of a given size~\cite{ZK07,AC08,CoEf15}.
The rough intuition (stated in terms of Karp's question) is that as the desired set size increases from $(1-\eps)\log_2 n$ to $(1+\eps)\log_2 n$, the collection of independent sets of $G_{n,1/2}$, as considered under a suitable and natural metric, abruptly transitions from a well-connected space into one with exponentially many well-separated pieces. After this transition, any algorithm for finding independent sets of the desired size (let alone colourings whose average part size is at least the desired set size) must ably navigate this shattered space. Thus $\log_2 n$ is considered an intuitive algorithmic barrier for the independent set problem in $G_{n,1/2}$, and analogously $n/\log_2 n$ is an algorithmic barrier for colouring $G_{n,1/2}$. A similar intuition should hold for binomial random graphs throughout the range of choices for the edge probability $p=p(n)$ satisfying $np = \Omega(1)$ and $p=o(1)$ and analogously also for random $\Delta$-regular graphs with $\Delta$ fixed, with thresholds at around $\frac{1}{p}\log (np)$ (or $np/\log (np)$) and $\frac{n}{\Delta}\log \Delta$ (or $\Delta/\log \Delta$), respectively. An affirmative answer to Karp's question, or its analogue for random regular graphs, would be considered an unexpected and sensational achievement. \cref{thm:Ckapplication} (just as does Theorem~\ref{thm:molloy} or~\ref{thm:AIS}) precisely matches this algorithmic barrier. In particular, by \cref{prop:regular} the random $\Delta$-regular graph provides examples of triangle-free, and so $F_k$-free, graphs $G$ such that  $\chi(G) \ge (1/2 - o(1))\Delta/\log\Delta$ as $\Delta\to\infty$. By comparison, \cref{thm:Ckapplication} with a choice of $\eps=\eps(\Delta) = o(1)$ (and so $t=\Delta^{o(1)}$) as $\Delta\to\infty$ (covering a much more general class of graphs) efficiently certifies $\chi(G)\le (1+o(1))\Delta/\log\Delta$.

The second and more concrete fact is that our framework (see \cref{sec:framework}) incorporates the quantitative
probabilistic property of local occupancy. 
We find asymptotically tight parameters for local occupancy in graphs
of maximum degree $\Delta$ where each vertex is contained in $\Delta^{o(1)}$ copies of $F_k$ (for any fixed $k\ge 3$) as $\Delta\to\infty$, see \cref{sec:tightness}.
Then to improve the leading order of the upper bound in \cref{thm:Ckapplication} for $t=\Delta^{o(1)}$ must
require novel techniques that use more sophisticated knowledge of the underlying structure of the graph.
The methods used to prove \cref{thm:molloy,thm:AIS,thm:Ckapplication} all hit the same obstruction, and our reduction of the problem to local occupancy suggests that surpassing these bounds will require a more global approach. 
Showing even only the existence of a colouring (so without requiring a polynomial-time construction) that betters these bounds by some constant factor, e.g.~via some global data that bypasses the local occupancy bottleneck, would be a breakthrough in classical Ramsey theory and graph colouring. 

An alluring feature of our work is the suggestion that local occupancy and the algorithmic barrier might be two sides of the same coin. Could it be the case that for locally sparse $n$-vertex graphs $G$, given $q$ such that there are enough independent sets of size at least $n/q$ in $G$ so that $q$-colouring can be performed efficiently, the collection of independent sets in the graph will be sufficiently rich and `well connected' to permit local occupancy with parameters that enable our method to show $\chi(G) \le q$?

\subsection{Organisation}

In~\cref{sec:framework} we introduce the main concepts in our framework and state the key results that establish its efficacy.
In \cref{sec:method} we give a general algorithm for graph colouring and analyse it with our framework.
In~\cref{sec:application} we verify the probabilistic information necessary to apply our framework in the case of graphs with few fans, completing the proof of \cref{thm:Ckapplication}, and in~\cref{sec:tightness} we discuss barriers to improving our results.
We defer several technical proofs to \cref{sec:mainproofs,sec:appproofs}.


\section{The framework}\label{sec:framework}

We introduce some extra notation for two concepts central to our framework.
First we do so for an important strengthened form of colouring, through which we prove all of our results. A \emph{$q$-list assignment} of $G$ is a function $L$ such that
$L(u)$, for each vertex~$u\in V(G)$, is a list of colours (natural numbers) of size $q$, and an \emph{$L$-colouring} of $G$ is a colouring with no monochromatic edges such that the colour of $u$ is a member of $L(u)$ for each vertex $u$.
Then the \emph{list chromatic number}~$\chi_\ell(G)$ of $G$ is the least integer $q$ such that every $q$-list assignment $L$ admits an $L$-colouring.
By taking $L(u)=\{1,\dotsc,q\}$ for each $u$ we see that $\chi(G)\le\chi_\ell(G)$ always.
Note that Johansson's and Molloy's bounds mentioned above were also shown in terms of $\chi_\ell$.

Second we write $\cI(G)$ for the set of independent sets in a graph $G$, and $\mu_{G,\lam}$ for the \emph{hard-core model on $G$ at fugacity $\lam$},
 the probability measure on $\cI(G)$ with
\[ \mu_{G,\lam}(I) \coloneqq \frac{\lam^{|I|}}{Z_G(\lam)}, \]
where $\lam>0$ is the \emph{fugacity parameter}, and $Z_G(\lam)\coloneqq \sum_{I\in \cI(G)}\lam^{|I|}$ is
the \emph{partition function}.  Its
\emph{occupancy fraction} is~$\EE |\bI|/|V(G)|$, where $\bI\sim \mu_{G,\lam}$, i.e.~the
expected fraction of the vertices in a random sample from $\mu_{G,\lam}$.  We frequently drop
subscripts when they are clear from context.

Our framework outputs a range of structural information for a graph with the verification of a condition in terms of the hard-core model, which we call {\em local occupancy}.
This systematic approach began with a number of previous works~\cite{DJPR17,DJPR18,MR02,DJKP18,DJKP18a} on occupancy fraction and fractional colouring. Here we focus on our framework's consequences for efficient (list) colouring, and refer the reader to the companion paper~\cite{DKPS20main} for more general structural implications.

\begin{definition}\label{def:localocc}
  We say that the hard-core model on a graph $G$ at fugacity $\lam$ has \emph{strong local
    $(\beta,\gam)$-occupancy} if, for every vertex $u\in V(G)$ and subgraph $F\subset G[N(u)]$ we have
  \[ \beta \frac{\lam}{1+\lam}\frac{1}{Z_F(\lam)} + \gam \frac{\lam Z'_F(\lam)}{Z_F(\lam)} \ge 1.  \]
\end{definition}

\noindent
One easily checks that this condition implies $\beta \Pr(u \in \bI) + \gam \EE|\bI \cap N(u)| \ge 1$ for any $u$, motivating the label.
The reason for including the adjective `strong' relates to some technical subtleties we discuss later as well as in~\cite{DKPS20main}. In an abbreviated form, our framework is as follows.

\begin{theorem}\label{thm:main}
  Let $G$ be a graph of maximum degree $\Delta\ge 2^6$ such that the hard-core model
  on $G$ at fugacity $\lam$ has strong local $(\beta,\gam)$-occupancy for
  some $\lam,\beta,\gam>0$. 
  Suppose there exists $\ell > 7\log\Delta$ such that for all vertices $u\in V(G)$ and subgraphs $F\subset G[N(u)]$ on at least $\ell/8$ vertices we have $Z_F(\lam) \ge 8\Delta^4$. Then the list-chromatic number of $G$ is at most $q$, where
          \begin{align*}
            q          & \coloneqq r\left(\beta + \gam\frac{\Delta}{r}\right)                     &
            \text{and} &                                                                          &
            r          & \coloneqq \frac{\lam}{1+\lam}\frac{\ell}{1-\sqrt{7(\log\Delta) / \ell}}.
          \end{align*}
\end{theorem}

\noindent
In this form, the framework essentially reduces the task of bounding the chromatic number from above to minimising $\beta+\gam \Delta/r$ subject to local $(\beta,\gam)$-occupancy, an optimisation which can be routinely performed to yield several other applications, cf.~\cite{DKPS20main}.
 We have given a conceptually elegant proof of \cref{thm:main} (in the style of Bernshteyn~\cite{Ber19}) in~\cite{DKPS20main}.
For polynomial-time constructions, here we need a more involved proof that requires some additional assumptions.

For an organic assimilation of the hard-core model in our arguments,
it will be helpful to represent list colourings through an auxiliary \emph{cover graph} as in the work of Dvo\v{r}\'{a}k and Postle~\cite{DP18}, on a stronger variant of list colouring called correspondence colouring.  
(This is a key insight.)

\begin{definition}
  Given a graph $G$, a \emph{cover} of $G$ is a pair $\sH = (L, H)$, consisting of a graph $H$ and a mapping $L \colon V(G) \to 2^{V(H)}$, satisfying the following requirements:
  \begin{enumerate}
    \item the sets $\{L(u) \,:\,u \in V(G)\}$ form a partition of $V(H)$;
    \item for every $u \in V(G)$, the graph $H[L(u)]$ is complete;
    \item if $E_H(L(u), L(v)) \neq \emptyset$, then either $u = v$ or $uv \in E(G)$;
    \item\label{item:matching} if $uv \in E(G)$, then $E_H(L(u), L(v))$ is a matching (possibly empty).
  \end{enumerate}
  A cover $\sH = (L, H)$ of $G$ is \emph{$q$-fold} if $|L(u)| = q$ for all $u \in V(G)$.
  An \emph{$\sH$-colouring} of $G$ is an independent set in $H$ of size $|V(G)|$.
\end{definition}

\noindent
Although covers as defined here capture a more general notion, most of our results here will remain restricted to list colouring. (We discuss the subtleties at the end.)
Given a~$q$-list assignment $\tilde L$ of $G$ we create a $q$-fold cover $\sH=(L,H)$ of $G$ such that $\sH$-colourings
of~$G$ correspond to $L$-colourings of $G$ by making the sets $L(u)$ formally disjoint copies of the lists
$\tilde L(u)$, and for every edge $uv\in E(G)$ adding an edge between $x\in L(u)$ and $y\in L(v)$ whenever $x$ and $y$ are two copies of the same colour. In
an attempt to avoid confusion we will refer to elements of the lists~$\tilde L(u)$ as natural numbers,
and we will refer to vertices of $H$ as colours.
We shorten the phrase `cover that arises from a list assignment' to `list-cover'.

We now state the extra assumptions needed for our framework to yield an efficient algorithm.

\begin{theorem}\label{thm:main:alg}
  Suppose that the conditions of \cref{thm:main} hold, and let~$n\coloneqq|V(G)|$.
  Suppose also that there is a class of graphs $\cC$ and an integer $t$ such that the following hold.
  \begin{enumerate}
    \item\label{itm:main-alg-remove}
          For each $u\in V(G)$, each induced subgraph $F\subset G[N(u)]$, and any list-cover $\sH'=(L',H')$ of $F$ with at most $\Delta$ colours in each list, we have a procedure $\remove(F,H')$, running in time
          $T_r\ge\Delta$, for finding a set $R$ of edges in $F$ such that $|R|\le t$ and the graph $\hat H$ obtained from $H'$ by removing any edge between $L'(v)$ and $L'(w)$ for $vw\in R$ satisfies
          $\hat H\in \cC$.
    \item\label{itm:main-alg-sample}
          For each $\hat H\in \cC$ we have a procedure $\sample(\hat H,\lam)$ for sampling from the
          hard-core model on $\hat H$ at fugacity $\lam$ in time $T_s$.
    \item\label{itm:main-alg-tvsell}
          The integer $t$ satisfies $0\le t \le \ell/40$.
  \end{enumerate}
  Then, for any $q$-list-assignment $\tilde L$ of $G$, there is a randomised algorithm that constructs,
  for any~$c\in(0,1)$, an $\tilde L$-colouring of $G$ as guaranteed by \cref{thm:main} in time
  \[ O\big((T_r+T_s)\Delta\log(\max(\ell/\lam,\ell))n+n^{1+c}\big) \]
  with probability at least $1-2/n^c$.
\end{theorem}

In general the graphs $F$ can have $\Delta$ vertices so the maximum of $T_s$ and $T_r$ could
be exponential in $\Delta$.  When $G$ has few copies of $F_k$ we show that these conditions hold with
$\cC$ the class of covers of $F_k$-free graphs with at most $\Delta$ colours in each list, and with $T_r$ and $T_s$ polynomial in $\Delta$, leading to the polynomial running time bound announced in \cref{thm:Ckapplication}.

\subsection{An algorithm for graph colouring}\label{sec:method}

First here is an overview of the two-phase method.
We define (precisely later) a \emph{flaw} for $u$ to capture the problem of having too few available colours or too much competition for available colours.
The first and foremost phase of the proof is that subject to the conditions of \cref{thm:main:alg},
a \emph{flawless partial colouring of $G$} can be found efficiently.

\begin{lemma}\label{lem:phase1:alg}
  Suppose that the conditions of \cref{thm:main:alg} hold, and let $\sH=(L,H)$ be a $q$-fold list-cover of $G$.
  Then there is an algorithm that constructs a flawless
  partial $\sH$-colouring of $G$ in time $O\big((T_r+T_s)\Delta\log(\max(\ell/\lam,\ell))\cdot n\big)$ with
  probability at least $1-2^{-n}$.
\end{lemma}

\noindent
The second phase is that a flawless partial colouring can be
efficiently completed to a list colouring of $G$. This `efficient finishing blow' is standard
and was established in earlier work~\cite{Mol19,AIS19}.

\begin{lemma}\label{lem:phase2}
  Suppose that the conditions of \cref{thm:main:alg} hold, $\sH=(L,H)$ is a $q$-fold list-cover of $G$,
  and let $\sigma$ be a flawless partial $\sH$-colouring of $G$.
  Then there is an algorithm that finds, for each~$c\in(0,1)$, an
  $\sH$-colouring of $G$ in time $O\left(n^{1+c}\right)$ with probability at least $1-1/n^c$.
\end{lemma}

Note that \cref{lem:phase1:alg} followed by \cref{lem:phase2} directly implies \cref{thm:main:alg}.

The algorithm for \cref{lem:phase2} selects a colour for the
remaining vertices uniformly at random, resampling if there are any conflicts.  For completeness we give
a sketch proof in \cref{sec:phase2}.

The algorithm for \cref{lem:phase1:alg} explores the space of partial colourings of $G$,
starting with a trivial colouring in which every vertex is \emph{coloured blank}.  We define
an order on flaws, and while the current partial colouring $\sigma$ is flawed we execute
a resampling action to address the least flaw present in $\sigma$ and move to a new partial colouring.
We give the proof of \cref{lem:phase1:alg} in the rest of this section with some details deferred to
\cref{sec:mainproofs}.  

\subsubsection{Notation}\label{sec:notation}

Given the setup of \cref{thm:main}, we work with a $q$-fold list-cover~$\sH=(L,H)$ of~$G$.
For a set $S\subset V(G)$ we write $L(S) \coloneqq \bigcup_{u\in S}L(u)$.  We refer to the
vertices of $H$ as colours, and write $H^*$ for the graph obtained from $H$ by removing all edges inside
the sets $L(u)$ for all~$u\in V(G)$.  Then for $u\in V(G)$ and $x\in L(u)$ we write $\deg^*_{\sH}(x)$ for the
degree in $H^*$ of a colour $x$, which is the number of colours on lists of neighbours of $u$ that
conflict with $x$.

Writing $\blank$ for a special blank colour, and borrowing from~\cite{AIS19},
a \emph{partial $\sH$-colouring~$\sigma$} of $G$ is a function from $V(G)$ to $\{\blank\}\cup
  V(H)\cup E(G)$ such that the following hold.
\begin{enumerate}
  \item For all $u\in V(G)$, either $\sigma(u)=\blank$, $\sigma(u)\in L(u)$, or $\sigma(u)=e\in E(G)$
        with $u\in e$.
  \item Restricting the image of~$\sigma$ to~$V(H)$ gives an independent set; 
        $\sigma(V(G))\cap V(H) \in \cI(H)$.
\end{enumerate}
We write $\Omega$ for the set of such partial $\sH$-colourings $\sigma$, and omit the prefix
$\sH$ when it is clear.

Given $\sigma\in\Omega$ we have \emph{blank vertices} $\bla(\sigma) \coloneqq \{u\in V(G) :
  \sigma(u)=\blank\}$ for which $\sigma(u)=\blank$, \emph{coloured vertices} $\col(\sigma) \coloneqq \{u\in V(G) :
  \sigma(u)\in L(u)\}$ for which $\sigma(u)\in L(u)$, and \emph{uncoloured vertices} $\unc(\sigma) \coloneqq
  \{u\in V(G) : \sigma(u)\in E(G)\}$ for which $\sigma(u)$ is an edge of $G$ containing $u$.  We also write
$\ind(\sigma) \coloneqq \sigma(V(G))\cap V(H) = \sigma(\col(\sigma))$ for the independent set in $H$ signified
by $\sigma$.

We also require some notation for the cover that remains on the blank vertices.  We write $G_\sigma
  \coloneqq G[\bla(\sigma)]$ for the subgraph of $G$ induced by $\bla(\sigma)$, and then
write $\sH_\sigma=(L_\sigma, H_\sigma)$ for the cover of $G_\sigma$ obtained by setting $L_\sigma(u)
  \coloneqq L(u) \setminus N_H[\ind(\sigma)]$ for~$u\in\bla(\sigma)$ and
$H_\sigma \coloneqq H[L_\sigma(\bla(\sigma))]$.
Note that $\sH_\sigma$ is a list-cover of~$G_\sigma$.  When there are no
uncoloured vertices, these definitions agree with those of~\cite{Ber19,DKPS20main} if $\sigma$ is identified with $\ind(\sigma)$.

To interpret this notation, note that the list $L_\sigma(u)$ contains the colours in $L(u)$ that do
not conflict with colours of the vertices in~$\col(\sigma)$.  This means that if $I\in \cI(H_\sigma)$, then $\ind(\sigma)\cup I\in \cI(H)$.  In particular, if $\sigma$ has no
uncoloured vertices, and if we can find $I\in\cI(H_\sigma)$ of size $|V(G_\sigma)|$
then $\ind(\sigma)\cup I$ is an $\sH$-colouring of $G$.  This is exactly how the two-phase method
proceeds.

\subsubsection{The flaws}

We define a flaw for each vertex $u$, writing
\begin{align*}
  B_u \coloneqq \{ \sigma \in \Omega : u\notin\col(\sigma)\text{ and either } & |L_{\sigma}(u)|<\ell,
  \text{ or }\exists x\in L_{\sigma}(u)\text{ with }\deg^*_{\sH_{\sigma}}(x)>\ell/8\}.
\end{align*}
We also define a flaw for each pair $(u, e)$ where $u\in e\in E(G)$, which represents the fact that~$u$ is
an uncoloured vertex with $\sigma(u)=e$, writing $U_u^e \coloneqq \{ \sigma\in\Omega : \sigma(u) = e \}$.
We write $F_B \coloneqq \{B_u : u\in V(G)\}$ and $F_U \coloneqq \{ U_u^e : u\in e\in E(G) \}$ so that
$F\coloneqq F_B\cup F_U$ is the set of all flaws.  Note that $|F_B|=|V(G)|=n$, and $|F_U|= 2|E(G)|\le
  \Delta n$.

It is important that we address the flaws in a sensible order, and any fixed order that puts every flaw
of the form $B_u$ before any flaw of the form $U_v^e$ suffices.  To be explicit, consider an arbitrary
ordering of the vertices, and the induced lexicographic ordering on edges where $uv$ is ordered according
to the pair $(u,v)$ with $u < v$.  We address the flaws consistent with the order that puts flaws of the
form $B_u$ first, ordered according to $u$, and then puts the $U_u^e$ ordered according to $u$ and then
$e$.

\subsubsection{The actions}

To address the flaw $B_u$ at state $\sigma$ we execute the action $\fix(u,\sigma)$ defined in
\cref{alg:addressb} in terms of the procedures $\remove$ and $\sample$ guaranteed by the assumptions of
\cref{thm:main:alg}.

\begin{algorithm}[ht]\footnotesize
  \caption{\label{alg:addressb}}
  \begin{algorithmic}[1]  
    \Procedure{$\fix$}{$u, \sigma$}
    \parState{\Let{} $\sigma'$ be obtained from $\sigma$ by setting $\sigma'(v) \coloneqq \blank$ for all $v\in N_G(u)\setminus\unc(\sigma)$, and $\sigma'(v)\coloneqq\sigma(v)$ otherwise}
    \State\Let{} $F\coloneqq G_{\sigma'}[N_G(u)]$, $H' \coloneqq H_{\sigma'}[L(N_G(u))]$, and let $\hat H \coloneqq \remove(F,H')$
    \Statex
    \parState{\Let{} $\bJ_0 \coloneqq \sample(\hat H,\lam)$, and \Let{}
      $\tau_0(v) \coloneqq \sigma'(v)$ unless~$v\in N_G(u)$ and~$\bJ_0\cap L(v)=\{y\}$, in which case
      $\tau_0(v) \coloneqq y$}\label{line:sample} \Statex
    \State \Let{} $i\coloneqq0$\label{line:unc:start}
    \While{$\bJ_i\notin \cI(H')$}
    \parState{\Let{} $vw$ be the lowest-indexed edge in $G[N(u)]$ for which $\bJ_i$ spans an edge of $H'$ going from $L(v)$ to~$L(w)$, and suppose that $v<w$}
    \parState{\Let{} $\tau_{i+1}$ be obtained from $\tau_i$ by setting $\tau_{i+1}(v) \coloneqq vw$ (uncolouring $v$) and setting $\tau_{i+1}$ to agree with~$\tau_i$ elsewhere}
    \State \Let{} $\bJ_{i+1} \coloneqq \bJ_i\setminus L(v)$
    \State \Increment{} $i$
    \EndWhile
    \Statex
    \State \Let{} $\tau\coloneqq\tau_i$ and \Let{} $\bJ \coloneqq \bJ_i$\label{line:unc:end}
    \State \Return{} $\tau$
    \EndProcedure
  \end{algorithmic}
\end{algorithm}

\cref{alg:addressb} has three distinct parts.  The first is some setup in which we define a partial
colouring $\sigma'$ by reassigning coloured vertices in $N_G(u)$ to $\blank$, which gives us an induced
subgraph~$F$ of~$G[N_G(u)]$ and a list-cover $H'\coloneqq H_{\sigma'}[L(N_G(u))]$ of $F$.  We then use the procedure $\remove$ to remove some edges from $F$ and any corresponding
edges in $H'$, which results in a cover $\hat H$ of $F$ in the class $\cC$.  The second is sampling
an independent set $\bJ_0$ in $\hat H$ and a partial colouring-like object $\tau_0$ corresponding to
$\bJ_0$.  Here we say partial colouring-like because although~$\bJ_0$ is independent in $\hat H$, it is
not necessarily independent in $H_{\sigma'}$ so $\tau_0$ is not necessarily a valid partial
$\sH$-colouring of $G$.  In the third part we iterate over a loop variable $i$ starting at $0$ and
uncolour vertices in $\bJ_i$ that participate in edges of $H$, making a
sequence of corresponding~$\tau_i$ as we go.  When the loop exits $\bJ_i$ is
independent in $H_{\sigma'}$, and so the final $\tau_i$ is a valid partial
colouring.

To address the flaw $U_u^e$ at state $\sigma$ we simply resample $\sigma(u)$ from the hard-core model as
follows.  Let $\sigma'$ be obtained from $\sigma$ by letting $\sigma'(v) \coloneqq \sigma(v)$ for $v\ne u$, and letting $\sigma'(u)\coloneqq\blank$ with probability $1/(1+|L_\sigma(u)|\lam)$, and otherwise setting
$\sigma'(u)$ to be a uniform colour from $L_\sigma(u)$.

\subsubsection{Proving termination}

In the analysis of the algorithm we discuss the transition probabilities induced by these actions,
writing $\rho_f(\sigma,\tau)$ for the probability that the final state is $\tau$ when addressing the
flaw~$f$ at state $\sigma$.  Let $\tilde\mu$ be the probability measure on $\Omega$ given by
\[ \tilde\mu(\sigma) \coloneqq
  \frac{\lam^{|\col(\sigma)|+|\unc(\sigma)|}}{\sum_{\tau\in\Omega}\lam^{|\col(\tau)|+|\unc(\tau)|}}, \]
and write $\tilde Z_H(\lam)$ for the denominator.
We note that $\tilde\mu$ is inspired by the hard-core model on~$H$; its definition is motivated by the
fact that creating an uncoloured vertex costs weight~$\lam$ when sampling $\bJ_0$ in the procedure
$\fix(u,\sigma)$.

Given subsets of flaws $S,S'\subset F$, we say that $S'$ \emph{covers} $S$ if\footnote{For readers
  familiar with the definitions of~\cite{AIS19}, this is because flaws in $F_U$ are \emph{primary} while
  those in $F_B$ are not.  We avoid making precise what primary means here, see~\cite{AIS19} for details.}
\[ S'\cap F_U = S\cap F_U \;\text{ and }\; S' \cap F_B \supset S\cap F_B.  \]
For any flaw $f$ and subset $S\subset F$ of flaws, we define\footnote{Note that this definition differs
  from that given in~\cite{AIS19} as we have additionally restricted to mappings~$\sigma$ such that
  $\rho_f(\sigma,\tau) > 0$.  The other states contribute zero to the charges we need to analyse so this
  change simply means that we avoid explicitly having to exclude such states in our analysis of charges.}
\begin{align*}
  \In_f^S(\tau) \coloneqq \{\sigma\in f :{} & \text{$\rho_f(\sigma,\tau)>0$ and the set of flaws}
  \\& \text{introduced by the transition $\sigma\to\tau$ covers $S$}\},
\end{align*}
and we define the \emph{charge} $c^S(f)$ to be
\begin{equation}\label{eq:charge}
  c^S(f) \coloneqq
  \max_{\tau\in\Omega}\Bigg\{\sum_{\sigma\in\In_f^S(\tau)}\frac{\tilde\mu(\sigma)}{\tilde\mu(\tau)}\rho_f(\sigma,\tau)\Bigg\},
\end{equation}
which represents a kind of compatibility between the measure $\tilde\mu$ and the transitions~$\rho_f$
induced by the actions for flaws~$f\in F$.

We can now state the main theorem of Achlioptas et al.~\cite[Theorem~2.4]{AIS19}, specialised to
our setting, to show that \cref{alg:addressb} terminates quickly with high
probability.  To compare with the original, more general statement, we point out that we use their
Remark~2.4 and that we start our algorithm in
the all-$\blank$ partial colouring which has measure $1/\tilde Z_H(\lam)$ and such that the only flaws present in the initial state are of the form $B_u$ (since there are no uncoloured vertices).

\begin{theorem}[Achlioptas, Iliopoulos, and Sinclair~\cite{AIS19}]\label{thm:algorithmicLLLL}
    If there exist positive numbers~${(\psi_f)}_{f\in F}$ such that for every $f\in F$ we have
  \[ \zeta_f \coloneqq \frac{1}{\psi_f}\sum_{S\subset F}c^S(f)\prod_{g\in S}\psi_g < 1, \]
  then for $s\ge 0$ \cref{alg:addressb} reaches a flawless state in $(T_0 + s)/\delta$ steps with
  probability at least $1-2^{-s}$, where $\delta \coloneqq 1 - \max_{f\in F}\{\zeta_f\}$ and
  \[ T_0 \coloneqq \log_2\tilde Z_H(\lam) + \sum_{u\in V(G)}\log_2(1+\psi_{B_u})
    + \log_2\left(\max_{S\subset F}\frac{1}{\prod_{f\in S}\psi_f}\right).  \]
\end{theorem}

\subsubsection{Bounding charges}

By design the compatibility between $\tilde\mu$ and our algorithm is good enough for the following result
to control the charges $c^S(B_u)$.

\begin{lemma}\label{lem:key}
  Suppose that the conditions of \cref{thm:main:alg} hold, and let $\sH=(L,H)$ be a $q$-fold list-cover
  of $G$.  Then for any partial $\sH$-colouring~$\sigma$ of~$G$
  and any $u\in V(G)$ such that $u\notin\col(\sigma)$, the following holds.  If~$\tau$ is the random
  partial colouring of $G$ that results from the procedure $\fix(B_u,\sigma)$, then $\Pr(\tau \in B_u)
    \le 1/(4\Delta^3)$.
\end{lemma}

A similar version of this is key to the lopsided local lemma formulation of our methods~\cite{DKPS20main}. To
obtain an algorithm we essentially take advantage of some subtle extra strength from the fact
that the lemma holds for any $\sigma$ rather than only when $\sigma$ has no uncoloured vertices and $\ind(\sigma)$ is sampled from the hard-core model on
$H$.
The proof of \cref{lem:key} is given in
\cref{sec:key}.

The following results comprise the bounds on charges we need to apply \cref{thm:algorithmicLLLL}, and the
proofs are in \cref{sec:chargeboundsB,sec:chargeboundsU}.  Let $S(U_u^e) \coloneqq \{ B_v : v\in N^2[u]
  \}$ and $S(B_u) \coloneqq \{ B_v : v\in N^3[u] \} \cup \{ U_v^{e} : v\in e\in E(G[N(u)]) \}$.  We will
see that these are the only flaws that addressing $U_u^e$ and $B_u$ can introduce, respectively.

\begin{lemma}\label{lem:chargeboundsB}
  For every vertex $u\in V(G)$, the following hold.
  \begin{enumerate}
    \item
          If $S\not\subset S(B_u)$ then $c^S(B_u) = 0$.
    \item
          If $S$ contains more than $t$ flaws of the form $U_v^{vw}$ with $v,w\in N(u)$ then $c^S(B_u) = 0$.
    \item
          $\max_{S\subset F}\{ c^S(B_u) \} \le 1/(4\Delta^3)$.
  \end{enumerate}
\end{lemma}

\begin{lemma}\label{lem:chargeboundsU}
  For every $u$ and $e$ such that $u\in e\in E(G)$, the following hold.
  \begin{enumerate}
    \item
          If $S\not\subset S(U_u^e)$ then $c^S(U_u^e)=0$.
    \item
          $\max_{S\subset F}\{ c^S(U_u^e) \} \le \lam/(1+\ell\lam)$.
  \end{enumerate}
\end{lemma}

\subsubsection{Finishing the proof}

We can now choose parameters~$\psi_f$ for~$f\in F$ such that the desired result follows from
\cref{thm:algorithmicLLLL}.  For this we take a positive real $\psi$ to be determined later and
set $\psi_f\coloneqq\psi/(4\Delta^3)$ for all $f\in F_B$ and~$\psi_f \coloneqq \psi\lam/(1+\ell\lam)$ for
all $f\in F_U$.

By Lemma~\ref{lem:chargeboundsB} we know that $c^S(B_u)=0$ unless all $B_v$ flaws in $S$ correspond to
vertices~$v$ in~$N^3[u]$, and there is a set $R\subset E(G[N(u)])$ of most $t$ edges such that if~$S$
contains a flaw of the form $U_v^{vw}$, then~$vw\in R$.  Since $|N^3[u]|\le\Delta^3$, we deduce from
\cref{lem:chargeboundsB} that for each~$u\in V(G)$,
setting
\[ \psi \coloneqq \frac{4(1+\ell\lam)}{1+\ell\lam +4t\lam}, \]
we have
\begin{align*}
  \frac{1}{\psi_{B_u}}\sum_{S\subset F}c^S(B_u)\prod_{g\in S}\psi_g
   & \le \frac{1}{\psi} \prod_{v\in N^3[u]}\left(1+\frac{\psi}{4\Delta^3}\right) \prod_{vw\in R}\left(1 + \frac{\psi\lam}{1+\ell\lam}\right)
  \\&\le \frac{1}{\psi}\exp\left(\frac{\psi}{4}\right)\exp\left(\frac{t\psi\lam}{1+\ell\lam}\right)
  = \frac{e}{4}\left(1+ \frac{4t\lam}{1+\ell\lam}\right) \le \frac{3}{4},
\end{align*}
because by assumption~$t\le\ell/40$, and hence~$4t\lam \le (1+\ell\lam)(3/e-1)$.

Similarly, by Lemma~\ref{lem:chargeboundsU} and the facts that $\Delta\ge 2$ and $|N^2[u]|\le 1+\Delta^2$, we have
for each pair~$(u,e)$ with~$u\in e\in E(G)$,
\begin{align*}
  \frac{1}{\psi_{U_u^e}}\sum_{S\subset F}c^S(U_u^e)\prod_{g\in S}\psi_g
   & \le \frac{1}{\psi} \prod_{v\in N^2[u]}\left(1+\frac{\psi}{4\Delta^3}\right)
  \le \frac{1}{\psi}\exp\left(\psi\frac{1+\Delta^2}{4\Delta^3}\right)                                  \\
   & = \frac{1}{4}\left(1+ \frac{4t\lam}{1+\ell\lam}\right)
  \exp\left(\frac{1+\Delta^2}{\Delta^3}\left(1-\frac{4t\lam}{1+\ell\lam+4t\lam}\right)\right)          \\
   & \le \frac{e^{(1+\Delta^2)/\Delta^3}}{4}\left(1+ \frac{4t\lam}{1+\ell\lam}\right) \le \frac{3}{4}.
\end{align*}
Hence we can apply \cref{thm:algorithmicLLLL} with parameters $\psi_{B_u}$ and $\psi_{U_u^e}$ such that
\begin{align*}
  \frac{e}{3\Delta^3}          & \le \psi_{B_u} = \frac{1}{\Delta^3}\cdot\frac{1+\ell\lam}{1+\ell\lam+4t\lam} \le \frac{1}{\Delta^3}, \\
  \frac{4e\lam}{3(1+\ell\lam)} & \le \psi_{U_u^e} = \frac{4\lam}{1+\ell\lam+4t\lam}\le \frac{4\lam}{1+\ell\lam},
\end{align*}
giving $\delta=1/4$ and
\[ T_0 \le
  \log_2\tilde Z_H(\lam) + n\log_2\left(1+\frac{1}{\Delta^3}\right) +
  n\log_2\left(\frac{3\Delta^3}{e}\right) + n\Delta\log_2\left(\frac{3(1+\ell\lam)}{4e\lam}\right).  \]
We have
\[ \tilde Z_H(\lam) = \sum_{\tau\in\Omega}\lam^{|\col(\tau)|+|\unc(\tau)|}
  \le {(1+2\Delta\lam)}^n, \]
because for each $u\in V(G)$ we can have either $\tau(u)=\blank$, which does not contribute to
the exponent of~$\lambda$, or $\tau(u)\in L(u)\cup \{e\in E(G) : u\in e\}$, which contributes~$1$.
There are at most $2\Delta$ choices in the latter case.  Then
\[ T_0 = O(n\log \Delta + n\log\lam + n\Delta\log(1/\lam) + n\Delta\log \ell).  \]
Therefore, $ T_0 = O(n\Delta\log(\ell) ) $ if~$\lam\ge1$
while $ T_0 = O(n\Delta\log( \ell/\lam ) )$ if~$\lam<1$.
Since $\delta=1/4$ and $T_0\ge n$, setting~$s\coloneqq n$ yields that the probability that the algorithm
finds a flawless partial colouring in at most $2T_0$ steps is at least $1-2^{-n}$.
Each step takes
time $O(T_r + T_s)$ because if we are addressing a flaw $f\in F_B$ then we execute action
$\fix(u,\sigma)$ which executes $\remove$ and $\sample$ once each in series, while the action to address a
flaw in $F_U$ is simply sampling from a distribution supported on at most $\Delta+1$ outcomes with
probabilities of the form $1/(1+y\lam)$ and $\lam/(1+y\lam)$ where $\ell\le y\le \Delta$.  This completes
the proof of \cref{lem:phase1:alg}, showing that the first, main phase of the algorithm works as desired.


\section{Application to graphs with few fans}\label{sec:application}

To prove~\cref{thm:Ckapplication} we must establish suitable strong local occupancy in graphs with few copies of $F_k$, and we must give suitable implementations of $\remove$ and $\sample$.
We start with the strong local occupancy, which relies on a \emph{maximum average degree} parameter
\[ \mad(G) \coloneqq \max_{\substack{F \subset G \\ |V(F)| \ge 1}}\left\{ \frac{2|E(F)|}{|V(F)|} \right\}, \]
in neighbourhoods.
We also write~$W$ for the (upper real branch of the) \emph{Lambert $W$-function} that is the inverse of
$x\mapsto xe^x$ defined on $\interval[co]{-1}{\infty}$.  We use the basic property that as $x\to\infty$
we have $W(x) = (1-o(1))\log x$, see e.g.~\cite{CGH+96}. 
The following result is proved in \cref{sec:appproofs} (see also~\cite{DJKP18a,DKPS20main}).

\begin{lemma}\label{lem:mad}
  Let~$a\ge0$ and $G$ be a graph such that $\mad(G[N(u)]) \le a$ for each~$u\in V(G)$. Then the
  following statements hold for any~$\lambda>0$.
  \begin{enumerate}
    \item\label{itm:mad-hcm} For any~$d>0$, there exist $\beta,\gam>0$ such that the hard-core model on~$G$ at fugacity~$\lam$ has strong local $(\beta,\gam)$-occupancy and
          \begin{align*}
              \beta+\gam d = \frac{1+\lam}{\lam} \frac{d{(1+\lam)}^{a}\log(1+\lam)}{W(d{(1+\lam)}^{a}\log(1+\lam))}.
          \end{align*}
    \item\label{itm:mad-largeZ} For any vertex~$u\in V(G)$ and any subgraph~$F$ of $G[N(u)]$ on $y$ vertices we
          have
          \begin{align*}
            \log Z_F(\lam) \ge y \log(1+\lam) \left(1-\frac{a}{2}\log(1+\lam)\right).
          \end{align*}
  \end{enumerate}
\end{lemma}

To apply the above result to $G$ as in \cref{thm:Ckapplication} we prove a suitable $\mad$ bound.

\begin{lemma}\label{lem:Ckmad}
  Let $u\in V(G)$ be contained in at most $t$ copies of the fan $F_k$. Then the average degree of any graph $F\subset G[N(u)]$ is at most $k-3 + \sqrt{2t}$.
\end{lemma}
\begin{proof}
  Let $F$ have $y$ vertices. We assert that the average degree of $F$ is at most
  \[ \min\left\{ y-1,\,k-3+ \frac{2t}{y}\right\} \le k-3 + \sqrt{2t}. \]
  The first bound is straightforward as there are at most $y-1$ possible neighbours for any vertex in
  $F$, and the second follows from a theorem of Erd\H{o}s and Gallai~\cite[Theorem~2.6]{EG59} that bounds the average degree of $P_{k-1}$-free graphs.
  By removing at most $t$ edges from $F$ we can remove all copies of $P_{k-1}$, and hence the resulting graph has at most $y(k-3)/2$ edges, which
  means $F$ has at most $y(k-3)/2+t$ edges.  The first expression in the assertion follows, and we
  consider the subcases $y\le \sqrt{2t}$ and $y>\sqrt{2t}$ to crudely bound from above the minimum.
\end{proof}

For the rest of this section let $G$ be as in \cref{thm:Ckapplication}.
Let~$\cF$ be the class of $P_{k-1}$-free graphs on at most $\Delta$ vertices and let $\cC$ be the class of list-covers of graphs in~$\cF$ with at most $\Delta$ colours in each list.
For any vertex $u\in V(G)$, subgraph~$F$ of~$G[N(u)]$, and list-cover $\sH=(L',H')$ of~$F$ with at most $\Delta$ colours in each list we can identify all copies of $P_{k-1}$ in $F$ in time $O(\Delta^k)$ by enumerating all ordered sets of $k-1$ vertices in $F$.
To implement $\remove$ we simply choose an arbitrary edge~$vw$ of each~$P_{k-1}$ found in this way and remove all edges from~$H'$ between~$L(v)$ and~$L(w)$.
For each of the at most $t\le \Delta^{2\eps}$ copies of~$P_{k-1}$ found, this removal takes time at most $\Delta^2$ so $\remove$ as in \cref{thm:main:alg} can be done in time $\Delta^{O(k)}$.

The following result which we prove in \cref{sec:appproofs} implies that we can implement $\sample$ on~$\cC$ in time ${(k\Delta)}^{O(k^3)}$.

\begin{restatable}{theorem}{sampling}\label{thm:sampling}
  Let $k\ge 3$ and $F$ be a $P_{k-1}$-free graph on $y$ vertices, and let $\hat \sH=(\hat L,\hat H)$ be a list-cover of $F$ with at most $q$ colours in each list.
  Then there is an absolute constant $c$ such that for any $\lam>0$ we can sample from the hard-core model on $\hat H$ in time $y^{3k^2}{(1+q)}^{k^3/2}{(ck)}^{k^3}$.
\end{restatable}

We can now finish the proof of \cref{thm:Ckapplication}.

\begin{proof}[Proof of \cref{thm:Ckapplication}]
  Fix an arbitrary vertex~$u\in V(G)$ and an arbitrary subgraph~$F\subset G[N(u)]$.
  By \cref{lem:Ckmad} we have $\mad(F)\le a \eqqcolon k-3+\sqrt{2t}$.
  It is convenient to exclude the case $a=0$ in the argument, which is one place the assumption $t\ge1/2$
  comes in useful, giving $a\ge 1$.

  We want to apply \Cref{thm:main:alg} to conclude the proof.
    Above we defined a class $\cC$, $\remove$, and $\sample$ such that \cref{itm:main-alg-remove,itm:main-alg-sample} hold with $T_r = \Delta^{O(k)}$ and $T_s = {(k\Delta)}^{O(k^3)}$.
  We now define the parameters so that the remaining requirements of \Cref{thm:main:alg} are satisfied:
  those are~$t\le \ell/40$ (\cref{itm:main-alg-tvsell}) and the hypothesis of 
  \Cref{thm:main}. In particular, we need to show that there is strong local $(\beta,\gamma)$-occupancy
  for the hard-core model on~$G$ at fugacity~$\lam$ for some positive reals~$\beta,\gamma$ and~$\lam$.
  To this end we use \Cref{lem:mad}, which will also provide the requirement on~$Z_F(\lam)$.  Indeed,
  with $a=k-3+\sqrt{2t}$, given any~$\lam>0$ and $\ell>7\log\Delta$, and with~$r$ and~$q$ as in
  \cref{thm:main}, \cref{itm:mad-hcm} of \cref{lem:mad} gives us $\beta$ and $\gam$ such that the
  hard-core model on $G$ at fugacity $\lam$ has strong local $(\beta,\gam)$-occupancy with
  \[ q= r\left(\beta + \gam \frac{\Delta}{r}\right)
    = \frac{1+\lam}{\lam}
    \frac{\Delta{(1+\lam)}^{a}\log(1+\lam)}{W(\Delta{(1+\lam)}^{a}\log(1+\lam)/r)}.  \]
  We set
  \begin{align*} \log(1+\lam) & \coloneqq \frac{1}{a\log(\Delta/\sqrt t)},
                 & \text{and}                                 &
                 & \ell                                       & \coloneqq \frac{40a}{\log(\Delta/\sqrt t)}{\left(\frac{\Delta}{\sqrt
      t}\right)}^{\frac{\eps}{1+\eps}}.
  \end{align*}
  First, recall that $t \le \Delta^{\frac{2\eps}{1+2\eps}}/{(\log\Delta)}^2$, and hence $40\cdot t/\ell
    = O\left({(\log\Delta)}^{-(1+\varepsilon)/(1+2\varepsilon)}\right)$ tends to~$0$ as~$\Delta$ goes to
  infinity, so~$t$ is indeed at most~$\ell/40$ if~$\Delta_0$ is large enough.

  Second, we must show that
  \begin{equation}
    \frac{\ell}{8} \ge \frac{\log(8\Delta^4)}{\log(1+\lam)\left(1 - \frac{a}{2}\log(1+\lam)\right)}
    = \frac{a\log(\Delta/\sqrt t)\log(8\Delta^4)}{1-\frac{1}{2\log(\Delta/\sqrt t)}},\label{eq-ellovereight}
  \end{equation}
  so that \cref{itm:mad-largeZ} of \cref{lem:mad} ensures that $Z_G(\lam)\ge 8\Delta^4$ for any subgraph
  $F$ of any $G[N(u)]$ on at least $\ell/8$ vertices.  The right-hand side of~\eqref{eq-ellovereight} is
  $O\left(a \log(\Delta/\sqrt t)\log\Delta\right)$ as $\Delta$ and hence~$\Delta/\sqrt t$ tend to
  infinity.  This is less than $\ell$ for large enough $\Delta_0$ in terms of $\eps$ because the bound
    $t\le\Delta^{\frac{2\eps}{1+2\eps}}/{(\log\Delta)}^2$ gives
  \[ \log(\Delta/\sqrt t) \ge \frac{1+\eps}{1+2\eps}\log\Delta - \log\log\Delta
    = \Omega(\log\Delta), \]
    so that $a \log(\Delta/\sqrt t)\log\Delta = O( a {(\log(\Delta/\sqrt t))}^2) = o(\ell)$.

  As $\Delta\to\infty$ we now have $\lam =o(1)$, $a\lam=o(1)$, $7\log\Delta=o(\ell)$, $a= O(\sqrt t)$,
  and
  \[ q
    \sim \frac{\Delta}{\log(\Delta/\ell)}
    \sim \frac{\Delta}{\log\left({(\Delta/\sqrt t)}^{1/(1+\eps)}\log(\Delta/\sqrt t)\right)}.  \]
  To obtain \cref{thm:Ckapplication}, note that for large enough $\Delta_0$ and $\Delta\ge\Delta_0$,
  \[ (1+\eps)\frac{\Delta}{\log(\Delta/\sqrt t)} \ge q , \]
  so the result follows from the application of \cref{thm:main:alg}.
\end{proof}


\section{Concluding remarks}\label{sec:conclusions}

In this section we compare the algorithmic framework given here with the more combinatorial treatment in the companion paper~\cite{DKPS20main}.
While in the introduction we deliberately omitted all mention of stronger graph colouring concepts, here we describe some subtleties related to several strengthenings handled within our framework, including local and correspondence colouring.

First our methods immediately generalise to a `local' formulation where each vertex is given a list of size depending primarily on $\deg(v)$ instead of $\Delta$.
See \cite{DJKP18,DKPS20main} for details, and in particular an interesting minimum list size phenomenon that arises.
Secondly, when applying Theorem~\ref{thm:main:alg} with a trivial implementation of $\remove$ that does nothing, our proof gives two extra properties.
When $\remove$ is trivial Theorem~\ref{thm:main:alg} works for correspondence colouring; one can dispense with the stated list-cover assumption in this case.
Moreover, if $\remove$ is trivial and we do have a list-cover then Theorem~\ref{thm:main:alg} only requires weak local occupancy: a variant of local occupancy that applies only to \emph{induced subgraphs} $F$ of $G[N(u)]$ instead of to arbitrary subgraphs.

This means that our methods give an algorithmic version of Bernshteyn's strengthening of Theorem~\ref{thm:molloy} to correspondence colouring~\cite{Ber19}, as well as the generalisation to $F_k$-free graphs for any $k\ge 3$ (recall that $F_3$ is a triangle and $F_3, C_k\subset F_k$). 
In cases where we have local occupancy (e.g.~one of the many settings covered in~\cite{DKPS20main}), the primary bottleneck for an efficient algorithm is $\sample$, and one of our contributions here is a polynomial-time implementation of $\sample$ that suffices for Theorem~\ref{thm:Ckfree}, or equivalently the case $t<1$ in Theorem~\ref{thm:Ckapplication}.
Without seeking an efficient algorithm, raising $t$ presents no serious challenge as our quantitative local occupancy guarantee degrades smoothly and slowly as $t$ increases. This is one of the key breakthroughs of our framework, see~\cite[Sec.~5]{DKPS20main}.  By contrast, the efficient implementation of $\sample$ is extremely fragile and fails completely at $t=1$. With the power of $\remove$ we can handle $t$ from $1$ up to the stated bound in Theorem~\ref{thm:Ckapplication}, but the cost of this is twofold. 
First we require a list-assignment as the analysis of the uncolouring steps depends crucially upon this. Second we require strong instead of weak local occupancy because removing edges can create arbitrary subgraphs of neighbourhoods $G[N(u)]$ at the relevant sampling stage, even with a list-assignment.

Here summarises some key differences among these strengthened colourings as treated both here and in~\cite{DKPS20main}.
Consider the problem of colouring a graph $G$ of maximum degree $\Delta$ in which each vertex is contained in at most $t$ $F_k$'s with $(1+\eps)\Delta/\log(\Delta/\sqrt t)$ colours, or lists of this size. 
Let $\Delta_0$ be large enough in terms of $\eps$ and $\Delta\ge\Delta_0$, and let $c=c(\eps)$ be a large enough constant. 
For usual graph colouring our \emph{existence} methods work up to $t\le\Delta^2/c$, but for the largest $t$ this requires a reduction that does not apply to list colouring. 
For list colouring we require $t \le \Delta^2/(\log\Delta)^{2/\eps}$, and our \emph{algorithmic} methods work for $t\le \Delta^{\frac{2\eps}{1+2\eps}}/(\log\Delta)^2$. 
For correspondence colouring our existence methods behave the same as for list colouring, but for an efficient algorithm we need $t<1$. 
It would be very interesting to learn if these differences are essential, or whether refined techniques can unify these results.

\section{Acknowledgements}

We thank Alistair Sinclair, Fotis Iliopoulos, and Charlie
Carlson for insightful discussions.

\appendix

\section{Tightness}\label{sec:tightness}

In this section we state a proposition indicating that for a large range of~$\lam$, the local
occupancy of \cref{lem:mad} is asymptotically best possible for $F_k$-free graphs.  First we note an extra (but foundational) component of our framework essentially originating in~\cite{DJPR18}, cf.~\cite{DKPS20main}.

\begin{theorem}\label{thm:localocc}
  Let $G$ be a graph of maximum degree $\Delta$ such that the hard-core model
  on $G$ at fugacity $\lam$ has (strong) local $(\beta,\gam)$-occupancy for
  some $\lam,\beta,\gam>0$. Then the occupancy fraction of $G$ at fugacity $\lam$ is at least $1/(\beta+\gam\Delta)$.
\end{theorem}

\noindent
Supposing that~$k$ is
a fixed integer greater than~$2$ and $\lam=o(1)$ such that $\lam\sqrt t=o(1)$ as $\Delta\to\infty$ we
know from \cref{lem:mad,thm:localocc} that for any graph~$G$ with maximum degree $\Delta$ in which each vertex is
the centre of at most $t$ copies of $F_k$, the occupancy fraction of $G$ at fugacity $\lam$ is at least
\[ (1-o(1))\frac{W(\Delta\lam)}{\Delta}.  \]
But the occupancy fraction is monotone increasing in~$\lam$ (see~\cite{DJPR18})
so this lower bounds holds for all larger values of~$\lam$ too. 
This is asymptotically tight at least for $\lam\le\Delta^{1+o(1)}$, but in fact no improvement to this lower bound is known for larger values of~$\lam$ even for the case of triangle-free graphs, that is
when~$k=3$.  In addition, when $t=\Delta^{o(1)}$ we can take $\lam = 1/\log\Delta$ and obtain a lower bound on the occupancy fraction which is~$(1-o(1))\log\Delta/\Delta$, and hence an improvement to the leading order for any $\lam$ larger than~$1/\log\Delta$ would immediately lead to an improvement to
Shearer's result that every $n$-vertex triangle-free graph of maximum degree $\Delta$ contains
an independent set of size at least $(1-o(1))n\log\Delta/\Delta$.

\begin{proposition}[\cite{DJPR18}]\label{prop:occtightness}
  Given $\eps>0$ there is $\Delta_0$ such that for all fixed $\Delta>\Delta_0$ and $\lam\le \Delta^{1+o(1)}$,
  there is a $\Delta$-regular $F_k$-free graph $G$ with occupancy fraction at most
  $(1+\eps)W(\Delta\lam)/\Delta$.
\end{proposition}

\section{Proofs for the main framework}\label{sec:mainproofs}

We require the following standard concentration inequality.
Given a probability space, the $\{0,1\}$-valued random variables~$\bX_1,\dotsc,\bX_n$ are
\emph{negatively correlated} if for each subset~$S$ of the set~$\{1,\dotsc,n\}$,
\[ \Pr\big(\bX_i=1, \forall i\in S\big)\le\prod_{i\in S}\Pr(\bX_i=1).  \]
\begin{lemma}[Panconesi and Srinivasan~\cite{PS97}]\label{lem:chernoff}
  Given a probability space, let~$\bX_1, \dotsc, \bX_n$ be $\{0,1\}$-valued random variables.
  Set~$\bX\coloneqq\sum_{i=1}^{n}\bX_i$ and~$\bY_i\coloneqq 1-\bX_i$ for each~$i\in\{1,\dotsc,n\}$.
  If the variables~$\bY_1,\dotsc,\bY_n$ are negatively correlated, then for any $\eta\in(0,1)$,
  \[ \Pr\big(\bX\le(1-\eta)\EE\bX\big) \le e^{-\eta^2\EE\bX/2}.  \]
\end{lemma}

\subsection{Proof of Lemma~\ref{lem:key}}\label{sec:key}

Recall that for \cref{lem:key} we assume the hypotheses of \cref{thm:main:alg} (and hence also
\cref{thm:main}), and we consider the transition from an arbitrary partial $\sH$-colouring $\sigma$
of $G$ induced by the action $\fix(B_u,\sigma)$ for some $u\notin\col(\sigma)$.  In what follows, we
write $\tau$ for the random state arising from this transition, and $\Pr$ and $\EE$ represent
probabilities and expectations over the randomness in the action $\fix(B_u,\sigma)$, respectively.

We first argue that $|L_{\tau}(u)|$ is large and concentrated around its expectation, and second that for
each $x\in L(u)$ the probability that $x$ is in~$L_{\tau}(u)$ and has large degree in $H^*_{\tau}$ is
small.  For this we require some additional notation.  We reuse the notation of the procedure $\fix$,
writing~$\sigma'$ for the partial $\sH$-colouring of $G$ obtained from $\sigma$ by setting
$\sigma'(v)\coloneqq\blank$ for all vertices~$v\in N_G(u)$, and $\sigma'(v)\coloneqq\sigma(v)$ otherwise.  We
also have the definitions $H'\coloneqq H_{\sigma'}[L(N_G(u))]$ and $\hat{H}\coloneqq\remove(G_{\sigma'}[N_G(u)], H')$, and~$\bJ_0$ is an independent set in~$\hat H$ sampled from
the hard-core model at fugacity~$\lam$.  If we write $L'(v) \coloneqq L(v)\cap V(H')$ then $\sH$ and
$\sH'\coloneqq(L',H')$ are list-covers of $G$ and $G_{\sigma'}[N_G(u)]$,
respectively.  This provides the additional structure that each colour~$x\in V(H)$ is considered a copy of
some natural number $c$ in the list-assignment, and $xy$ is an edge of $H^*$ if and only if $x\in L(u)$
and $y\in L(v)$ are copies of the same natural number and~$uv\in E(G)$.

For each~$x\in L(u)$, let $\Lam_x$ be the \emph{layer} of $x$, given by $\Lam_x\coloneqq N_{H^*_{\sigma'}}(x)$.  This
consists of the colours in~$L_{\sigma'}(N_G(u))$ that conflict with $x$, and hence $\Lam_x$ consists of
every colour $y\in V(H')$ that is a copy of a fixed natural number, written~$c$.  So for distinct $x,y\in
  L(u)$ the layers $\Lam_x$ and~$\Lam_y$ are necessarily disjoint.  The fact that $\sH'$ is a list-cover of
$G[N_G(u)]$ means that every edge leaving~$\Lam_x$ in~$H'$ joins
two colours belonging to some set~$L(v)$ with~$v\in N_G(u)$, which means there are no edges between one layer $\Lam_x$ and another $\Lam_y$ and facilitates the analysis of each layer separately. 
The set $\bJ_0$ sampled in step~\ref{line:sample} of the procedure $\fix$ is not
necessarily independent in $H'$ and the uncolouring steps~\ref{line:unc:start}--\ref{line:unc:end}
yield the subset $\bJ\subset \bJ_0$ which is independent in $H'$.  For any $x\in L(u)$ we note that
$\bJ\setminus \Lam_x$ depends only on $\bJ_0\setminus \Lam_x$ and not on $\bJ_0\cap \Lam_x$ because of
the above property of edges leaving $\Lam_x$.  To see this, observe that every uncolouring step is due to
some tuple $(v_1,v_2,y_1,y_2)$ such that $\{v_1,v_2\}\subset N_G(u)$, $v_1v_2\in E(G)$, $y_1\in
  L(v_1)\cap \bJ_0$, $y_2\in L(v_2)\cap \bJ_0$, and $y_1y_2\in E(H')$.  But then $y_1$ and $y_2$ must be a copy
of the same natural number $c'$ and hence either both are in~$\Lam_x$ or neither one is in~$\Lam_x$.
This means the uncolouring steps due to edges of $H'$ inside $\Lam_x$ are independent of the other
uncolouring steps, and hence $\bJ\setminus \Lam_x$ depends only on $\bJ_0\setminus \Lam_x$.

We now show some key properties of how $\bJ_0$ is distributed on the sets
$\Lam_x$, and how this affects~$L_\tau(u)$ and the flaw $B_u$.  Let us write
$\bU_0(x)$ for the set of vertices obtained by revealing~${\bJ_0\setminus
\Lam_x}$ and taking those vertices in $\Lam_x$ that in the graph $\hat H$ are
not adjacent to any vertex of $\bJ_0\setminus \Lam_x$.  That is, $\bU_0(x)
\coloneqq \Lam_x \setminus N_{\hat H}(\bJ_0 \setminus \Lam_x)$.  Then write
$\bF_0(x) \coloneqq \hat H [\bU_0(x)]$.  By the spatial Markov property of the
hard-core model, $\bJ_0\cap \Lam_x$ is distributed according to the hard-core
model on the graph $\bF_0(x)$ at fugacity $\lam$.  It is important to observe
that $\hat H[\Lam_x]$ is isomorphic to a subgraph of $G[N(u)]$, as is
$\bF_0(x)$, so the assumptions of \cref{thm:main} give that $\bF_0(x)$ has
strong local $(\beta,\gam)$-occupancy.

The above definitions deal with the cover $\hat H$ in which we sample, but we must also deal
with the original cover $H'$.  To this end, we analogously write $\bU(x)$ for the set of vertices
obtained by revealing $\bJ\setminus \Lam_x$ and taking those vertices in $\Lam_x$ that in the graph $H'$
are not adjacent to any vertex of $\bJ\setminus \Lam_x$.  That is, $\bU(x) \coloneqq \Lam_x
  \setminus N_{H'}(\bJ \setminus \Lam_x)$.  We also write $\bF(x) \coloneqq H' [\bU(x)]$.
We now note the following facts that hold for all $x\in L(u)$.
\begin{description}
  \item[\namedlabel{itm:xinLtau}{Fact~1}]
        $x\in L_\tau(u)$ if and only if $\bJ_0\cap \Lam_x = \varnothing$.
  \item[\namedlabel{itm:Jsubset}{Fact~2}]
        $\bJ\subset \bJ_0$ and hence $\bU(x) \supset \bU_0(x)$.
  \item[\namedlabel{itm:degstarx}{Fact~3}]
        If $x\in L_\tau(u)$ then $\deg^*_{\sH_\tau}(x) = |\bU_0(x)|$.
\end{description}

To see~\ref{itm:xinLtau}, observe that $x\in L_\tau(u)$ if and only if for every $v\in\col(\tau)\cap
  N_G(u)$ we have $x\tau(v)\notin H$.  That is, if $x$ is a copy of the natural number $c$ then $x\in
  L_\tau(u)$ if and only if no neighbour of $u$ is coloured with a copy of $c$ under $\tau$.  This clearly
holds if $\bJ_0\cap \Lam_x = \varnothing$, and also if $\bJ_0\cap \Lam_x \ne \varnothing$ then at least
one neighbour of $u$ is coloured with a copy of $c$.  But in each uncolouring step we only uncolour one
end of a monochromatic edge, so that if $\bJ_0\cap \Lam_x \ne \varnothing$ we must also have
$\bJ\cap\Lam_x \ne \varnothing$.

For~\ref{itm:Jsubset} note that in each uncolouring step we remove a vertex from $\bJ_i$ to form
$\bJ_{i+1}$ and so $\bJ$ (which is the set after all uncolouring steps) is a subset of $\bJ_0$.

Finally, for~\ref{itm:degstarx} suppose that $x$ is a copy of the natural number $c$ and note that
$\deg^*_{\sH_\tau}(x)$ counts the number of colours in $L_{\tau}(N_G(u))$ that conflict with $x$, or
equivalently the number of neighbours $v\in N_G(u)$ that are coloured blank by $\tau$ and that have a
copy of $c$ present in their list~$L_\tau(v)$.  A colour in $L_{\tau}(N_G(u))$ that conflicts with $x$
must be present in $\bU(x)$, and in order that $x\in L_\tau(u)$ we must have $\bJ\cap \bU(x)=\varnothing$
by~\ref{itm:xinLtau} and~\ref{itm:Jsubset}.  This means that every colour~$y\in \bU(x)$ such that $y\in L(v)$ for some
$v\in\bla(\tau)$ contributes to $\deg^*_{\sH_\tau}(x)$.  Then it suffices to show that for $y\in L(v)$ we
have the following two properties
\begin{align}
   & y\in \bU_0(x) \Rightarrow v\in\bla(\tau),\label{eq:U0bla}                    \\
   & y\in\bU(x)\setminus\bU_0(x)\Rightarrow v\in\unc(\tau).\label{eq:UminusU0unc}
\end{align}
This crucially exploits the fact that $\sH'$ is a list-cover.  To prove~\eqref{eq:U0bla},
observe that if $y\in\bU_0(x)$ then no neighbour of $y$ in~$\hat H$ belongs to $\bJ_0\setminus\Lam_x$,
and since $x\in L_\tau(u)$ we have $\bJ_0\cap \Lam_x=\varnothing$ by~\ref{itm:xinLtau}.  Together
these facts mean that $v$ is neither coloured nor uncoloured, and hence belongs to~$\bla(\tau)$
as required.  To
prove~\eqref{eq:UminusU0unc}, suppose that $y\in \bU(x)\setminus \bU_0(x)$.  Then there must be some
$z\in (\bJ_0\setminus \Lam_x) \setminus (\bJ\setminus\Lam_x)$ such that $yz$ is an edge of~$H'$.  As
reported earlier, since $\sH'$ is a list-cover we know that every edge leaving $\Lam_x$ is
inside some set~$L(w)$, and as~$y\in L(v)$ we deduce that $z\in L(v)$.  But then the reason we have
$z\in (\bJ_0\setminus \Lam_x) \setminus (\bJ\setminus\Lam_x)$ is because $z$ had been coloured and was
then removed from some $\bJ_i$ in the creation of the next~$\bJ_{i+1}$, and so
necessarily~$v\in\unc(\tau)$ as required.

We are now ready to show that $|L_\tau(u)|$ is likely to be at least~$\ell$.

\begin{lemma}\label{lem:LIu}
  Writing
  \[ m \coloneqq
    \frac{1+\lam}{\beta\lam}(q-\gam\Delta) \qquad\text{and}\qquad \eta\coloneqq\sqrt{7(\log\Delta)/\ell}, \]
  we have $\Pr\big(|L_{\tau}(u)| \le (1-\eta)m\big) \le e^{-\eta^2m/2} \le 1/(8\Delta^3)$.
\end{lemma}
\begin{proof}

  We first note that $e^{-\eta^2m/2} \le 1/(8\Delta^3)$ holds because the parameter choices in
  \cref{thm:main} give
  \[ m = \frac{\ell}{1-\eta} \ge \ell = \frac{7\log\Delta}{\eta^2} \ge \frac{6\log(2\Delta)}{\eta^2}.  \]

  Now by~\ref{itm:xinLtau}, we have $x\in L_\tau(u)$ if and only if $\bJ_0\cap\Lam_x=\varnothing$.
  As reported earlier,~$\bJ_0\cap\Lam_x$ is distributed according to the hard-core model on a graph
  $\bF_0(x)$ that is isomorphic to a subgraph of $G[N_G(u)]$.  Then $\bF_0(x)$ has strong local
  $(\beta,\gam)$-occupancy and so
  \begin{equation}\label{eq:layer:localocc}
    \beta \frac{\lam}{1+\lam}\Pr(x\in L_{\tau}(u))+ \gam\EE|\bJ_0\cap \Lam_x| \ge 1,
  \end{equation}
  because elementary calculations with the hard-core model (see e.g.~\cite{DJPR18}) now give
  \[ \Pr(\bJ_0\cap\Lam_x=\varnothing) = \frac{1}{Z_{\bF_0(x)}(\lam)}
    \qquad\text{and}\qquad
    \EE|\bJ_0\cap \Lam_x| = \frac{\lam Z'_{\bF_0(x)}(\lam)}{Z_{\bF_0(x)}(\lam)}.  \]
  We sum~\eqref{eq:layer:localocc} over all $q$ colours $x\in L(u)$ to obtain
  \[ q
    \le \beta \frac{\lam}{1+\lam}\EE|L_{\tau}(u)| + \gam\sum_{x\in L(u)}\EE|\bJ_0\cap \Lam_x|
    \le \beta \frac{\lam}{1+\lam}\EE|L_{\tau}(u)| + \gam\Delta,  \]
  where the last inequality holds because $\sum_{x\in L(u)}\EE|\bJ_0\cap \Lam_x|=\EE\sum_{x\in
      L(u)}|\bJ_0\cap \Lam_x|$ and every neighbour of~$u$ contributes at most~$1$ to the sum
  as $|\bJ_0\cap L(v)|\le1$ for every vertex~$v$.
  Rearranging immediately yields $\EE|L_{\tau}(u)|\ge m$, and the result will follow from an
  application of \cref{lem:chernoff}.

  For this application, note that $\EE|L_{\tau}(u)|$ is a sum over $x\in L(u)$ of the indicator
  variables~$\bX_x$ for the events~$\{ \bJ_0\cap \Lam_x = \varnothing \}$.  We can apply
  \cref{lem:chernoff} if we show that the random variables~$\bY_x\coloneqq 1-\bX_x$ are negatively
  correlated.  This correlation was shown formally by Bernshteyn~\cite{Ber19} in the triangle-free
  case, and is somewhat intuitive here.  Consider the random set~$\bJ_0$.
  Given~$x\in L(u)$, if $\bJ_0 \cap \Lam_x = \varnothing$ then no colours conflicting with~$x$ are chosen
  for vertices in~$N(u)$.  This makes other colours more likely to be chosen, such as those which
  conflict with $x'\in L(u)\setminus\{x\}$.  We repeat Bernshteyn's argument for completeness.

  It is enough to show that for all~$x\in L(u)$ and~$Y\subset L(u)\setminus\{x\}$ we have
  \[ \Pr\big(x\notin L_{\tau}(u) \bigm| Y \cap L_{\tau}(u) = \varnothing\big) \le
    \Pr(x\not\in L_{\tau}(u)), \]
  which is equivalent to
  \[ \Pr\big(Y \cap L_{\tau}(u) = \varnothing \bigm| x\in L_{\tau}(u) \big)
    \ge \Pr(Y \cap L_{\tau}(u)=\varnothing), \]
  which we can write (using~\ref{itm:xinLtau}) as
  \begin{align*}
    \Pr\big(\bJ_0\cap \Lam_y \neq \varnothing \text{ for all } y\in Y \bigm| \bJ_0\cap{} & \Lam_x =\varnothing\big)
    \\&\ge \Pr\big(\bJ_0\cap \Lam_y \neq \varnothing \text{ for all } y\in Y\big).
  \end{align*}
  This holds because the layers $\Lam_z$ for $z\in L(u)$ are pairwise disjoint.
\end{proof}

We now prove a result designed to handle the degree condition in the flaw~$B_u$.

\begin{lemma}\label{lem:deg}
  For any~$x\in L(u)$, writing
  \[ D_x \coloneqq \big\{x\in L_{\tau}(u) \text{ and } \deg^*_{\sH_{\tau}}(x) > \ell/8\big\}, \]
  we have $\Pr(D_x) \le 1/(8q\cdot\Delta^3)$.
\end{lemma}
\begin{proof}
  By~\ref{itm:degstarx} we know that if $x\in L_\tau(u)$ then $\deg^*_{\sH_\tau}(x) = |\bU_0(x)|$, so it
  suffices to show whenever $|\bU_0(x)|>\ell/8$ that
  \[ \Pr(x\in L_{\tau}(u)) \le \frac{1}{8q\Delta^3}.  \]
  As already reported, by~\ref{itm:xinLtau} (and an elementary property of the hard-core model) we
  have $\Pr(x\in L_{\tau}(u)) = 1/Z_{\bF_0(x)}(\lam)$.  We recall that $\bF_0(\lam)$ is isomorphic to a
  subgraph of~$G[N(u)]$, so that the upper bound on $\Pr(x\in L_{\tau}(u))$ follows directly from the
  assumptions on $\ell$ and~$Z_F(\lam)$ for $F\subset G[N(u)]$ stated in
  \cref{thm:main}.  There we assume that for all such~$F$ on at least $\ell/8$ vertices we have
  $Z_F(\lam)\ge 8\Delta^4$. Consequently, noticing that without loss of generality we have $q\le\Delta$
  (for otherwise a greedy argument finds any $q$-list colouring of $G$), we infer that
  \[ \Pr(x\in L_{\tau}(u)\bigm| |\bU_0(x)|>\ell)\le\frac{1}{8\Delta^4}\le\frac{1}{8q\Delta^3}.\]
  The result follows.
\end{proof}

The combination of \cref{lem:LIu,lem:deg} completes the proof of~\cref{lem:key}.

\subsection{Proof of Lemma~\ref{lem:chargeboundsB}}\label{sec:chargeboundsB}
\begin{proof}[Proof of \cref{lem:chargeboundsB}]
  Addressing $B_u$ by executing $\fix(u,\sigma)$ only modifies $\sigma$ on $N(u)$ and hence can only
  introduce flaws $B_v$ for $v\in N^3[u]$, or $U_v^{vw}$ for $v,w\in N(u)$ and $vw\in E(G[N(u)])$.

  In the procedure $\fix$ we can only introduce at most $t$ flaws of the form $U_v^{vw}$ since the
  procedure $\remove$ removes at most $t$ edges from $F$ that can be monochromatic, and each
  monochromatic edge leads to at most one uncolouring.

  If addressing $B_u$ results in $\tau$ with positive probability then the previous state $\sigma$ must
  agree with $\tau$ outside of $N(u)$, and inside $N(u)$ we must have $\unc(\tau)\supset\unc(\sigma)$ as
  no uncoloured vertices are coloured by the procedure $\fix$.  In fact, by the definition of covering
  for sets of flaws, when bounding $c^S(B_u)$ we can restrict our attention to triples
  $(S,\sigma,\tau)$ such that $\tau$ and $\sigma$ agree outside $N(u)$ and
  \[ \unc(\tau) = \unc(\sigma) \cup \{ v\in N(u) : U_v^e\in S \text{ for some $e\in E(G)$}\}.  \]
  That is, we can restrict attention to triples $(S,\sigma,\tau)$ where $S$ carries the information
  necessary to deduce $\unc(\sigma)$ from $\unc(\tau)$.  To this end, write
  \[ \nu(S) \coloneqq \{ v\in N(u) : \text{$U_v^e\in S$ for some $e\in E(G)$}\}, \]
  so that $\nu(S)$ is the set of uncoloured vertices present in $\tau$ but not $\sigma$ for any $\sigma\in\In^S_{B_u}(\tau)$.
  We now assert that for such triples $(S,\sigma,\tau)$ we have
  \[ \rho_{B_u}(\sigma,\tau) = \frac{\lam^{|\col(\tau)\cap N(u)| + |\nu(S)|}}{Z_{\hat{H}}(\lam)}, \]
  where $\hat{H}$ is the graph $\hat{H}$ appearing in $\fix(u,\sigma)$; repeated here for convenience,
  $\hat{H}$ is constructed as follows.  Let $\sigma'$ be obtained from $\sigma$ by setting $\sigma'(v)
    \coloneqq \blank$ for all $v\in N(u)\setminus\unc(\sigma)$, and $\sigma'(v)\coloneqq\sigma(v)$
  otherwise, and then let $F\coloneqq G_{\sigma'}[N_G(u)]$ and $H'\coloneqq H_{\sigma'}[L(N_G(u))]$.
  Then $\hat{H}\coloneqq\remove(F,H')$.  Note that~$\hat{H}$ does not depend on $\sigma$ in the sense
  that when $(S,\sigma,\tau)$ are as above we can construct~$\hat{H}$ from~$S$ and~$\tau$ alone.  Here
  we crucially exploit the definition of covering for sets of flaws.  The key point is that $\sigma'$
  as above can be constructed from $S$ and $\tau$ because the uncoloured vertices of $\sigma$ are given
  by $\unc(\sigma)=\unc(\tau)\setminus\nu(S)$.

  The assertion now follows from the definition of $\fix(u,\sigma)$ because from $(S,\tau)$ we can
  recover the size of the independent set $\bJ$ sampled in the procedure when $\tau=\fix(u,\sigma)$.
    More accurately, we can determine $\col(\bJ)$, which must be the disjoint union
    of~${\col(\tau)\cap N(u)}$ and the set $\nu(S)$ of vertices uncoloured during the execution
    of~$\fix(u,\sigma)$.  Given~$\col(\bJ)$ the uncolouring steps are deterministic so we have the
    required expression for~$\rho_{B_u}(\sigma,\tau)$ by the definition of the hard-core model.  We
    also see that
  \[ \frac{\tilde\mu(\sigma)}{\tilde\mu(\tau)} = \frac{\lam^{|\col(\sigma)\cap N(u)|+|\unc(\sigma)\cap
        N(u)|}}{\lam^{|\col(\tau)\cap N(u)|+|\unc(\tau)\cap N(u)|}} = \frac{\lam^{|\col(\sigma)\cap
        N(u)|}}{\lam^{|\col(\tau)\cap N(u)|+|\nu(S)|}}, \]
  which holds by the definition of $\nu$ and because $\sigma$ and $\tau$ only differ in $N(u)$.  Then
  we have shown that given $S$ and $\tau$, for any $\sigma\in\In^S_{B_u}(\tau)$ we have
  \begin{equation}\label{eq:chargeterm}
    \frac{\tilde\mu(\sigma)}{\tilde\mu(\tau)}\rho_{B_u}(\sigma,\tau)
    = \frac{\lam^{|\col(\sigma)\cap N(u)|}}{Z_{\hat{H}}(\lam)},
  \end{equation}
  where $\hat{H}$ can be obtained from $S$ and $\tau$ alone.

  Turning to the charge $c^S(B_u)$, we have
  \begin{align}
    c^S(B_u)
     & =  \max_{\tau\in\Omega}\left\{\sum_{\sigma\in\In_{B_u}^S(\tau)}\frac{\tilde\mu(\sigma)}{\tilde\mu(\tau)}\rho_{B_u}(\sigma,\tau)\right\}
    =\max_{\substack{\tau\in\Omega : \\ u\notin \col(\tau)}}\left\{\sum_{\sigma\in\In_{B_u}^S(\tau)}\frac{\lam^{|\col(\sigma)\cap N(u)|}}{Z_{\hat{H}}(\lam)}\right\},\label{eq:cBubound}
  \end{align}
  where $\hat{H}$ can be computed from $S$ and $\tau$ as above.  The restriction to $\tau$ such that
  $u\notin\col(\tau)$ is valid because for any $\tau\in\Omega$ with $u\in\col(\tau)$ we have
  $\In^S_{B_u}(\tau)=\varnothing$ since $\sigma\in B_u$ means that $u\notin\col(\sigma)$, but then
  $\rho_{B_u}(\sigma,\tau)=0$ because the procedure $\fix(u,\sigma)$ does not alter $\sigma(u)$.

  Nearing conclusion, we now argue that Lemma~\ref{lem:key} implies $c^S(B_u) \le 1/(4\Delta^3)$.
  This holds because given $S$ and $\tau$ we can construct the $\hat{H}$ occurring
  in~\eqref{eq:cBubound}, and note that it is a bona fide cover of some induced subgraph $F\subset
    G[N(u)]$.  That is, there is a cover $\hat\sH=(L,\hat H)$ of $G$ with the following two properties.
  \begin{enumerate}
    \item One obtains $\hat H$ from $H$ by removing edges in $H^*$, and hence $\hat\sH$ satisfies the hypotheses of \cref{thm:main}.
    \item There is a partial colouring independent set $\hat J\in\cI(\hat H)$ such that $\hat{H} = \hat H_{\hat J}[L(N_G(u))]$.
  \end{enumerate}
  Then the sum in~\eqref{eq:cBubound} over states $\sigma\in\In^S_{B_u}(\tau) \subset B_u$ can be interpreted as
  the probability that when $\col(\sigma)\cap N(v)$ is a random independent set from the hard-core model
  on $\hat{H}$ at fugacity $\lam$, we have $\sigma\in \In^S_{B_u}(\tau)$.  Since
  $\In^S_{B_u}(\tau)\subset B_u$, we can bound this from above by the probability that $\sigma$ belongs
  to~$B_u$ given this random experiment; and Lemma~\ref{lem:key} shows this probability to be at most~$1/(4\Delta^3)$.
\end{proof}

\subsection{Proof of Lemma~\ref{lem:chargeboundsU}}\label{sec:chargeboundsU}
\begin{proof}[Proof of \cref{lem:chargeboundsU}]
  Addressing $U_u^e$ at state $\sigma$ can only introduce flaws of the form $B_v$ for $v$ in~$N^2[u]$
  because it causes no uncolouring, and only affects $\sigma(u)$.

  If addressing $U_u^e$ results in $\tau$ with positive probability then the previous state $\sigma$ must
  be obtained from $\tau$ by setting $\tau(u)=e$.  That is, for $S\subset F$, the set
  $\In^S_{U_u^e}(\tau)$ is either empty, or contains exactly one state $\sigma$.  In the former case
  the charge is zero; and in the latter case, if $u\in\bla(\tau)$ then $\tilde\mu(\sigma)/\tilde\mu(\tau)=\lam$ and
  $\rho_{U_u^e}(\sigma,\tau) = 1/(1+ |L_\sigma(u)|\lam)$, and if $u\in\col(\tau)$ then
  $\tilde\mu(\sigma)/\tilde\mu(\tau)=1$ and $\rho_{U_u^e}(\sigma,\tau) = \lam/(1+ |L_\sigma(u)|\lam)$, hence both
  possibilities yield that
  \[ \frac{\tilde\mu(\sigma)}{\tilde\mu(\tau)}\rho_{U_u^e}(\sigma,\tau) = \frac{\lam}{1+ |L_\sigma(u)|\lam}.  \]
  Then $c^S(U_u^e) \le \lam/(1+ \ell\lam)$ because the ordering on flaws ensures that $\sigma\notin B_u$
    whenever we are addressing $U_u^e$, and hence $u$ has at least $\ell$ available colours.
\end{proof}

\subsection{Proof of Lemma~\ref{lem:phase2}}\label{sec:phase2}

We can prove \cref{lem:phase2} with the local lemma of Moser and Tardos~\cite{MT10}.  Let $\sigma$ be a
flawless partial colouring of $G$, and if necessary remove colours from~$H_\sigma$ such that
$|L_\sigma(u)|=\ell$ for all $u\in G_\sigma$.  We let $I'\subset V(H_\sigma)$ be obtained by choosing for
each remaining blank vertex~$u\in \bla(\sigma)$ a uniform random colour from $L_\sigma(u)$.  Since
$\sigma$ is flawless there are no uncoloured vertices, so the conclusion follows if $I'$ is independent.
Then for this algorithmic local lemma application we have a flaw $A_{xy}$ for each edge $xy$ of
$H_\sigma^*$, and $I'\in A_{xy}$ if and only if $\{x,y\}\subset I'$.  To address $A_{xy}$ we resample
$I'\cap L(u)$ and $I'\cap L(v)$ independently, uniformly at random.

Suppose that $x\in L(u)$, $y\in L(v)$, and $xy\in E(H^*_\sigma)$.  Then $A_{xy}$ is independent of any
$A_{x'y'}$ such that $x'$ and $y'$ are not in $L_\sigma(u)\cup L_\sigma(v)$.  In this simple setting the
full charge machinery of \cref{thm:algorithmicLLLL} is not necessary, and it suffices to use the
result of Moser and Tardos~\cite[Thm.~1.2]{MT10}. We use the slightly stronger version found in
Iliopoulos' doctoral thesis~\cite[Thm.~3.9, p.~17]{Ili}: it implies the original formulation by
setting~$x_A\coloneqq\frac{\psi_A}{1+\psi_A}\in(0,1)$.

\begin{theorem}[Iliopoulos~\cite{Ili}]
  Let $\cP$ be a finite set of mutually independent random variables, and let $\cA$ be a finite set of
  events determined by these variables.  Let $\Gam(A) \coloneqq \{ B\in \cA : \text{$A$ and $B$ are
      dependent}\}$, and let $\mu$ be the probability measure that results from sampling the variables
  $\cP$.  If for each $A\in\cA$ there exists a positive real $\psi_A$ such that
  \[ \frac{\mu(A)}{\psi_A} \sum_{S\subset \Gam(A)}\prod_{B\in S}\psi_B \le 1, \]
  then there is a randomised algorithm that finds an assignment to $\cP$ violating none of the events
  in $\cA$ with expected number of resamplings $\sum_{a\in\cA}\psi_A$.
\end{theorem}

Note that we have $A\in\Gam(A)$ since without loss of generality $\mu(A)\in (0,1)$ and hence $A$ depends
on itself.

In our setting, the variables in~$\cP$ are $I'\cap L(u)$ for $u\in V(G_\sigma)$ and the set of events is
$\cA \coloneqq \{ A_{xy} : xy\in H^*_\sigma\}$.
Supposing that $x\in L(u)$, $y\in L(v)$, and $xy\in E(H^*_\sigma)$, the fact that~$\sigma$ is flawless
implies that
\[ \mu(A_{xy}) = \frac{1}{|L(u)||L(v)|} = \ell^{-2}, \]
and also that
\[ |\Gam(A_{xy})|\le \sum_{x'\in L(u)}\deg^*_{H_\sigma}(x')
  + \sum_{y'\in L(v)}\deg^*_{H_\sigma}(y')\le \ell^2/4.  \]
Hence it suffices to take $\psi \coloneqq 4\ell^{-2}$ to have, for each~$A\in\cA$,
\[ \frac{\mu(A)}{\psi_A} \sum_{S\subset \Gam(A)}\prod_{B\in S}\psi_B
  \le \frac14\cdot{(1+4/\ell^2)}^{\ell^2/4} \le\frac{e}{4}<1. \]
The expected number of resamplings is then at most~$n/4$, because~$H^*_\sigma$, having no more
than~$n\cdot\ell$ vertices and maximum degree no more than~$\ell/8$, contains at most~$n\ell^2/16$ edges.
Consequently, Markov's inequality implies that, for any real~$c\in(0,1)$, the algorithm succeeds in
$O\left(n^{1+c}\right)$ resamplings with probability at least~$1-1/n^c$. Since a resampling is done in
constant time, this constitutes a proof of \cref{lem:phase2}.

\section{Proofs for the application}\label{sec:appproofs}

\subsection{Proof of Lemma~\ref{lem:mad}}

The proof of \cref{lem:mad} relies on the following elementary lemma, which already
appeared~\cite{DJKP18a}, though we give the short proof here for completeness.

\begin{lemma}\label{lem:sparsehcm}
  For any graph~$F$ on~$y$ vertices with positive average degree at most~$a$,
  \begin{align*}
    \frac{\lam Z_F'(\lam)}{Z_F(\lam)}
      & \ge \frac{\lam}{1+\lam}y{(1+\lam)}^{-a},&\text{and}&&
    \log Z_F(\lam)
      & \ge \frac{y}{a}\left(1-{(1+\lam)}^{-a}\right).
  \end{align*}
\end{lemma}
\begin{proof}
  Let~$\bS$ be a random independent set from the hard-core model at fugacity~$\lam$ on~$F$.
  First, for any $u\in V(F)$,
  \[ \Pr(u\in\bS) = \frac{\lam}{1+\lam}\Pr(\bS\cap N(u)
    = \varnothing)
    \ge \frac{\lam}{1+\lam}{(1+\lam)}^{-\deg(u)}, \]
  because the spatial Markov property gives that $\bS\cap N(u)$ is a random independent set drawn from
  the hard-core model on the subgraph $F[N(u)]$ induced by the externally uncovered neighbours of $u$.
  That is, when $\bU\coloneqq N(u)\setminus N(\bS\setminus N(u))$ is the set obtained by revealing
  $\bS\setminus N(u)$ and removing from~$N(u)$ any vertex with a neighbour in $\bS\setminus N(u)$, we
  write $\bF_{N(u)} \coloneqq F[\bU]$, and then $\bS\cap N(u)$ is distributed according to the
  hard-core model on $\bF_{N(u)}$ at fugacity~$\lam$.  The final inequality comes from the fact that
    any realisation of $\bF_{N(u)}$ has $Z_{\bF_{N(u)}}(\lam)\le {(1+\lam)}^{\deg(u)}$.  The lemma now
  follows by convexity:
  \begin{align*}
      \EE|\bS| = \sum_{u\in V(F)}\Pr(u\in\bS) & \ge \frac{\lam}{1+\lam}\sum_{u\in V(F)}{(1+\lam)}^{-\deg(u)}
                                              \ge \frac{\lam}{1+\lam}y{(1+\lam)}^{-a},
  \end{align*}
  and since
  \[ \EE|\bS| = \frac{\lam Z_F'(\lam)}{Z_F(\lam)} = \lam \frac{\partial}{\partial \lam} \log Z_F(\lam), \]
  integrating this bound gives the required lower bound on~$\log Z_F(\lam)$.
\end{proof}

\begin{proof}[Proof of \cref{lem:mad}]
  Let $u$ be an arbitrary vertex of $G$, and suppose that $F\subset G[N(u)]$ has $y$ vertices.
  By assumption we know that~$F$ has average degree at most~$a$.
    If $a=0$ then $F$ contains no edges and $Z_F(\lam) = {(1+\lam)}^y$, otherwise by \cref{lem:sparsehcm} we have
    \[ \log Z_F(\lam) \ge \frac{y}{a}(1-{(1+\lam)}^{-a}) \ge y\log(1+\lam)\left(1-\frac{a}{2}\log(1+\lam)\right), \]
  which establishes \cref{itm:mad-largeZ}.
    For \cref{itm:mad-hcm} we note that $Z_F(\lam)\le {(1+\lam)}^y$ and hence by \cref{lem:sparsehcm} we have
  \[ \beta\frac{\lam}{1+\lam}\frac{1}{Z_F(\lam)} + \gam \frac{\lam Z'_F(\lam)}{Z_F(\lam)}\ge
    \frac{\lam}{1+\lam}\Big(\beta{(1+\lam)}^{-y} + \gam y{(1+\lam)}^{-a}\Big), \]
  and we define the right-hand side to be~$g(y)$.

  The function $g$ is strictly convex with a stationary minimum at
  \[ y^* \coloneqq a + \frac{\log\left(\frac{\beta}{\gam}\log(1+\lam)\right)}{\log(1+\lam)}, \]
  and if we set~$g(y^*)=1$ for strong local $(\beta,\gam)$-occupancy, and solve for~$\beta$ we obtain
          \begin{align*}
              \beta & \coloneqq \frac{\gam{(1+\lam)}^{\frac{{(1+\lam)}^{1+a}}{\gam\lam}-a}}{e\log(1+\lam)}.
          \end{align*}
  Then the function~$\beta+\gam d$ is strictly convex
  in~$\gam$, and the unique minimiser is attained when 
            \begin{align*}
                \gam  & \coloneqq \frac{1+\lam}{\lam}\frac{{(1+\lam)}^{a}\log(1+\lam)}{1+W(d{(1+\lam)}^{a}\log(1+\lam))}.
          \end{align*}
  One checks that, indeed, setting~$\beta$
    and~$\gam$ to the announced values, and writing~$D$ for the expression~$d{(1+\lam)}^{a}\log(1+\lam)$, we have
  \[ \beta = \frac{1+\lambda}{\lambda}\cdot\frac{e^{W(D)}}{1+W(D)}
    = \frac{1+\lambda}{\lambda}\cdot\frac{D}{W(D)\cdot(1+W(D))}, \]
  noticing that~${(1+\lambda)}^{1+a}/(\gamma\lambda)=(1+W(D))/\log(1+\lambda)$,
  and hence
  \[ \beta+\gam d=\frac{1+\lam}{\lam}\left(\frac{D}{W(D)(1+W(D))}+\frac{D}{1+W(D)}\right)
    =\frac{1+\lam}{\lam} \frac{D}{W(D)}, \]
  as announced.  Furthermore, $y^* = W(D)/\log(1+\lambda)$, and hence indeed
  \begin{align*}
      g(y^*) & = \frac{D\cdot{(1+\lambda)}^{-W(D)/\log(1+\lambda)}}{W(D)(1+W(D))} + \frac{W(D)}{1+W(D)} \\
           & = \frac{1}{1+W(D)}+\frac{W(D)}{1+W(D)}=1.\qedhere
  \end{align*}
\end{proof}

\subsection{Proof of Theorem~\ref{thm:sampling}}

As is well known (and straightforward to realise), in a connected graph any two longest paths must have a common vertex, which
we state as follows.

\begin{lemma}\label{lem:longestPintersect}
  Let $k\ge 3$ and $F=(V,E)$ be a connected $P_{k-1}$-free graph. If $P$ is the vertex set of a longest path in $F$, then $F[V\setminus P]$ is $P_{k-2}$-free.
\end{lemma}

We now prove a warm-up to \cref{thm:sampling} that shows we can calculate the partition function and certain useful probabilities efficiently.

\begin{lemma}\label{lem:ZandP}
  Let $k\ge 3$ and $F$ be a $P_{k-1}$-free graph on $y$ vertices, and let $\sH=(L,H)$ be a cover of $F$ with at most $q$ colours in each list.
  Then there is an absolute constant $c$ such that for any positive~$\lam$ we can evaluate $Z_H(\lam)$ in time $y^{3k}{(1+q)}^{k^2/2}{(ck)}^{k^2}$.

  Further, let $F_1,\dotsc, F_r$ be the connected components of $F$.
  Write $H_i = H[L(V(F_i))]$ for the covers of the components of $F$ given by $H$, and let $\bI$ be a random independent set from the hard-core model at fugacity $\lam$ on $H$.
  Then in time $y^{3k}{(1+q)}^{k^2/2}{(ck)}^{k^2}$ we can compute for every~$i$ in~$[r]$ a longest path $P_i$ in $F_i$, and the probabilities $\Pr(\bI\cap L(P_i) = J_i)$ for every independent set $J_i\subset L(P_i)$.
\end{lemma}

\begin{proof}
  In the proof we write $c_1,c_2,\dotsc$ for some unspecified absolute constants, and let $f(y,q,k)$ be the upper bound on running time that we wish to calculate.

  If $k=3$ then $F$ is edgeless and hence
  \[ Z_H(\lam) = \prod_{v\in V(F)}(1+|L(v)|\lam), \]
  which is computable in time $c_1 y$.
  This gives the base case $f(y,q,3)=c_1y$.

  For $k>3$, we note that $F$ has at most $y$ connected components and deal with each one separately.
  The function $Z_H(\lam)$ is multiplicative over the induced covers of the components of~$F$, so
  \[ Z_H(\lam) = \prod_{i=1}^r Z_{H_i}(\lam). \]

  Let $F'=(V',E')$ be a component of $F$, and write $H'=H[L(V')]$ for the induced cover of~$F'$.
  Let $P=\{v_1,\dotsc,v_{k-2}\}$ be the vertex set of a longest path $F'$,
  which we can find in time $c_2 (k-2)!2^{k-2}y$ by executing for each~$j$ from~$k-1$ to~$2$ a fixed-parameter algorithm that finds a path of length $j$ (if one exists) in time $O(j!2^j y)$, see~\cite{Bod93}.

  Since $F'$ is connected, Lemma~\ref{lem:longestPintersect} implies that $F'[V'\setminus P]$ is $P_{k-2}$-free.
  We use the fact that the sum over independent sets $I$ in $H'$ that gives $Z_{H'}(\lam)$ can be split into terms according to~$I\cap L(P)$.
  Let $H^P = H'[L(P)]$ and $H^{V'\setminus P} = H'[L(V'\setminus P)]$. Then
  \[ Z_{H'}(\lam) = \sum_{J\in\cI(H^P)} \lam^{|J|} \cdot Z_{H^{V'\setminus P}_J}(\lam). \]
  It is possible to iterate over the required $J\in\cI(H^P)$ in time $c_3{(k-2)}^2{(1+q)}^{k-2}$ by iterating over all ${(1+q)}^{k-2}$ sets $J$ with $|J\cap L(v_i)|\le 1$ and checking each for independence.
    The independence check takes time $O\big({(k-2)}^2\big)$ as we must verify that each pair of vertices in $J$ is absent from~$E(H)$.
  Then constructing $H^{V\setminus P}_J$ can be done in time $c_4(k-2)y$ because $|J|\le k-2$ and each colour in $J$ can conflict with at most $y$ colours which are removed from $H^{V\setminus P}$ to form~$H^{V\setminus P}_J$.
  Now $Z_{H^{V\setminus P}_J}(\lam)$ can be computed in time $f(y,q,k-1)$ by induction, so we can compute $Z_{H'}(\lam)$ in time
  \[ c_2(k-2)!2^{k-2} y\cdot c_3{(k-2)}^2{(1+q)}^{k-2}\cdot c_4(k-2)y\cdot f(y,q,k-1). \]

  Since there are at most $y$ components of $F$, we have the recurrence
  \[ f(y,q,k) = y\cdot c_2(k-2)!2^{k-2} y\cdot c_3{(k-2)}^2{(1+q)}^{k-2}\cdot c_4ky\cdot f(y,q,k-1). \]
  With the base case of $f(y,q,3)=c_1y$  we have
  \[ f(y,q,k) \le c_1 y {(c_2c_3c_4y^3)}^{k-3}{(2(1+q))}^{k(k-3)/2}{((k-2)!)}^4\cdot \prod_{j=1}^{k-3}j!,\]
  which for a large enough constant $c$ is at most $y^{3k}{(1+q)}^{k^2/2}{(ck)}^{k^2}$.

  For the final statement, observe that in the above argument we compute for each $i\in [r]$ an evaluation of $Z_{H_{J_i}[V(F_i)\setminus P_i]}$ for some longest path $P_i$ in $F_i$, and each independent set $J_i\subset L(P_i)$ along the way to computing $Z_{H_i}(\lam)$.
  But
  \[ \Pr(\bI\cap L(P_i) = J_i) = \frac{\lam^{|J_i|} Z_{H_{J_i}[V(F_i)\setminus P_i]}(\lam)}{Z_{H_i}(\lam)}, \]
  so with some straightforward extra bookkeeping we have the required probabilities.
\end{proof}

With this result we can give the required sampling algorithm, which is restated here.

\sampling*

\begin{proof}
  We write $\bI$ for a random independent set from the hard-core model on $\hat H$.

  If $k=3$ then $F$ is edgeless and it suffices to sample independently for each vertex $v\in V$.
  With probability $1/(1+|\hat L(v)|\lam)$ take $\bI\cap \hat L(v)=\varnothing$, otherwise let $\bI\cap \hat L(v)$ be a uniform random element of $\hat L(v)$.
  This can be done in time $O(qy)$ provided sampling from a biased coin takes time $O(1)$ and sampling uniformly from a list of length $q$ takes time $O(q)$.

  If $k>3$ then let $F_1,\dotsc,F_r$ be the components of $F$.
  By Lemma~\ref{lem:ZandP} we can compute in time $y^{3k}{(1+q)}^{k^2/2}{(ck)}^{k^2}$ a longest path $P_i$ in $F_i$ and the probabilities $\Pr(\bI \cap \hat L(P_i) = J_i)$ for all independent sets $J_i\subset \hat L(P_i)$ and for all $i$.
  Hence we can sample $\bI \cap \hat L(P_i)$ for all $i$ in this time.
    Then we can construct $H'_i=\hat H_{\bI \cap \hat L(P_i)}[\hat L(V(F_i)\setminus P_i)]$ for each $i$ in time $O(rky)$, and use the fact that $\bI\cap \hat L(V(F_i)\setminus P_i)$ is distributed according to the hard-core model at fugacity $\lam$ on $H'_i$.

  With this scheme it is straightforward to show by induction that the time taken is at most
  \[ y^{3k^2}{(1+q)}^{k^3/2}{(ck)}^{k^3}. \qedhere \]
\end{proof}

\bibliographystyle{habbrv}
\bibliography{hcm_cs}

\end{document}